\documentclass{article}

\usepackage{arxiv}
\usepackage{style}

\usepackage[utf8]{inputenc} 
\usepackage[T1]{fontenc}    
\usepackage{hyperref}       
\usepackage{url}            
\usepackage{booktabs}       
\usepackage{amsfonts}       
\usepackage{nicefrac}       
\usepackage{microtype}      
\usepackage{lipsum}
\usepackage{graphicx}
\graphicspath{ {./images/} }

\usepackage{amsmath,amssymb,amsfonts, amsthm}
\usepackage{graphicx}
\usepackage{xcolor}
\usepackage{pgfplots}
\usepackage{array} 
\usepackage{comment}
\usepackage{multirow}
\usepackage{hyperref}

\usepackage{algorithm}
\usepackage{algpseudocode}

\usepackage{rotating} 

\newtheoremstyle{spaced}
  {10pt}   
  {10pt}   
  {\itshape} 
  {}       
  {\bfseries} 
  {.}      
  { }      
  {}       

\theoremstyle{spaced}

\newtheorem{definition}{Definition}
\newtheorem{remark}{Remark}

\newtheorem{proposition}{Proposition}
\newtheorem{myexample}{Example}

\title{A decomposition approach for large virtual network embedding}

\author{
  Amal Benhamiche \\
  Orange Research \\
  F-92320 Châtillon, France \\  \texttt{amal.benhamiche@orange.com} \\
   \And
    Pierre Fouilhoux \\
   Laboratoire d'Informatique de Paris Nord, LIPN, CNRS \\
  Université Sorbonne Paris Nord\\
  F-93430 Villetaneuse, France   \\
  \texttt{pierre.fouilhoux@lipn.univ-paris13.fr} \\
  \And
   Lucas Létocart \\
   Laboratoire d'Informatique de Paris Nord, LIPN, CNRS \\
  Université Sorbonne Paris Nord\\
  F-93430 Villetaneuse, France   \\
  \texttt{lucas.letocart@lipn.univ-paris13.fr} \\
  \And
  Nancy Perrot \\
  Orange Research \\
  F-92320 Châtillon, France \\
  \texttt{nancy.perrot@orange.com} \\
   \And
  Alexis Schneider \\
  Orange Research \\
  F-92320 Châtillon, France \\
  \texttt{alexis.schneider@orange.com} \\
  \& Laboratoire d'Informatique de Paris Nord, LIPN, CNRS \\
  Université Sorbonne Paris Nord\\
  F-93430 Villetaneuse, France   \\
  \texttt{alexis.schneider@lipn.univ-paris13.fr} \\
}

\begin{document}

\maketitle
\begin{abstract}
Virtual Network Embedding (VNE) 
is the core combinatorial problem of Network Slicing, a 5G technology which enables telecommunication operators to propose diverse service-dedicated virtual networks, embedded onto a common substrate network. VNE asks for a minimum-cost mapping  of a virtual network on a substrate network, encompassing simultaneous node placement and edge routing decisions.
On a benchmark of large virtual networks with realistic topologies we compiled, the state-of-the art heuristics often provide expensive solutions, or even fail to find a solution when resources are sparse.
We introduce a new integer linear formulation together with a decomposition scheme based on an automatic partition of the virtual network. This results in a column generation approach 
whose pricing problems are also VNE problems. This method allows to compute better lower bounds than state-of-the-art methods.
Finally, we devise an efficient Price-and-Branch heuristic for large instances.\end{abstract}

\keywords{Virtual Network Embedding \and Network Slicing \and Mathematical Programming \and  Decomposition \and Column Generation}

\section{Introduction}

In the early 2000s, \textit{network virtualization} was proposed to overcome the Internet's growing ossification and to foster innovations in networks \cite{anderson2005overcoming, turner2005diversifying, chowdhury2010survey}. 
In this concept, the physical network hosts multiple \textit{virtual networks}, each with its own architecture, protocols, and users. 
A virtual network is a logical topology composed of virtual nodes, representing, for example, a Virtual Network Function (VNF), a data center or a container, and virtual edges, which imply bandwidth requirements between those nodes. 
The \textit{infrastructure provider} is responsible for deploying and managing these virtual networks on the underlying physical infrastructure. 
One of the key challenges in network virtualization is the assignment of virtual elements to the physical infrastructure, which considers simultaneous placement (of the virtual nodes) and routing (of the virtual edges) decisions. The underlying combinatorial problem is known as the Virtual Network Embedding (VNE). \\~\\

Network virtualization introduces an abstraction layer between service providers and the physical network, enabling flexible, efficient, and automated network management.
It was later adopted as a cornerstone of 5G architecture under the name \textit{network slicing} \cite{alliance2016description, 3gpp2018slicing}. 
A \textit{slice} corresponds to an end-to-end virtual network tailored to a specific service or family of services, with specific QoS requirements such as low latency for Ultra-Reliable Low-Latency Communication (URLLC) applications or high bandwidth for Enhanced Mobile Broadband (eMBB). 
Network slicing is pivotal to enable, within 5G, emerging technologies such as autonomous vehicles or Internet of Things \cite{foukas2017network, ordonez2017network, afolabi2018network}. 
With the large-scale deployment of 5G Standalone (SA), real-world slices have begun to appear: T-Mobile has introduced a dedicated first-responder slice in the United States \cite{Tmobile2025priority}, while O2 Telefónica and Siemens deployed an industrial slice for the water sector in Germany \cite{Siemens2025water}. 
These slices cover large portions of the physical network, spanning multiple cities. 
Because they serve critical services and essential infrastructure (likewise emergency response or water supply), their implementation onto the physical network must be carefully planned to ensure robustness and long-term stability. 
In this paper, our aim is to propose efficient algorithms for such slice cases.

\subsection{Problem Statement}

The Virtual Network Embedding Problem (VNE) can be defined as follows. 
We consider simple, connected, undirected graphs. Let us denote the virtual network as $\Graph_r = (V_r, E_r)$ with $n_r$ nodes; and the substrate network as $\Graph_s = (V_s, E_s)$ with $n_s$ nodes. 
Allocating the substrate resources to the virtual demands is to find a \textit{mapping} (also called an embedding) as follows:

\begin{definition}
    A mapping $m = (m_V, m_E)$ of $\Graph_r$ on $\Graph_s$ is a pair of functions where: \\
    $\bullet$ $m_V: V_r \rightarrow V_s$ is called the \textit{node placement};\\
    $\bullet$ $m_E: E_r \rightarrow P_s$, is called the \textit{edge routing}, where $P_s$ is the set of loop-free paths of $\Graph_s$ and for each $\ebar = (\ubar,\vbar) \in E_r$, $m_E(\ebar)$ is a path of $\Graph_s$ whose endpoints are $m_V(\ubar)$ and $m_V(\vbar)$. 
\end{definition}

In this work, we consider a \textit{one-to-one} node placement: each virtual node has to be placed on a different substrate node.
However, our formulations and algorithms can be easily adapted to the \textit{many-to-one} node placement case.
A substrate node $u \in V_s$ (resp. a substrate edge $e \in E_s$) has an integer capacity  $c_u$ (resp. $c_e$).
Similarly, a virtual node $\ubar \in V_r$ (resp. virtual edge $\ebar \in E_r$) has an integer demand $d_{\ubar}$ (resp. $d_{\ebar}$).
A mapping is said to be \textit{feasible} when the capacity constraints are satisfied, i.e., the sum of the demands of virtual nodes (resp., edges) being placed on a substrate node (resp., being routed on a substrate edge) is less than or equal to the capacity of this node (resp., edge). \\

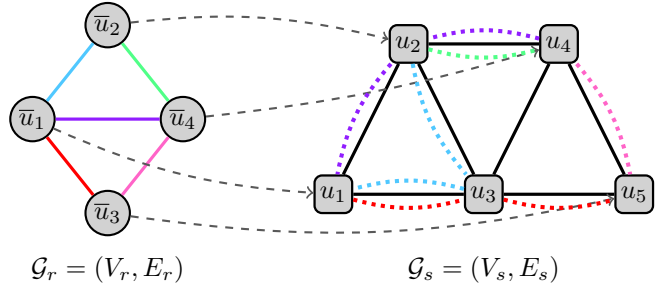
\begin{figure}
    \centering
    \begin{tikzpicture}
    \begin{scope}
        \node at (2,-2) {$\Graph_r = (V_r, E_r)$};
        \node[vnode] (v1) at (1,0){$\ubar_1$};
        \node[vnode] (v2) at (2,1.25){$\ubar_2$};
        \node[vnode] (v3) at (2,-1.25){$\ubar_3$};
        \node[vnode] (v4) at (3,0){$\ubar_4$};
    \end{scope}
    \draw[networkedge, blue] (v1) -- (v2);
    \draw[networkedge, purple] (v1) -- (v4);
    \draw[networkedge, green] (v2) -- (v4);
    \draw[networkedge, red] (v1) -- (v3);
    \draw[networkedge, pink] (v3) -- (v4);

    \begin{scope}[xshift=5.cm]
        \node at (2,-2) {$\Graph_s = (V_s, E_s)$};
        \node[snode] (u1) at (0,-1){$u_1$};
        \node[snode] (u2) at (1,1){$u_2$};
        \node[snode] (u3) at (2,-1){$u_3$};
        \node[snode] (u4) at (3,1){$u_4$};
        \node[snode] (u5) at (4,-1){$u_5$};
    \end{scope}
    \draw[networkedge] (u1) -- (u2);
    \draw[networkedge] (u1) -- (u3);
    \draw[networkedge] (u2) -- (u4);
    \draw[networkedge] (u2) -- (u3);
    \draw[networkedge] (u3) -- (u4);
    \draw[networkedge] (u4) -- (u5);
    \draw[networkedge] (u3) -- (u5);

    \draw[placementedge] (v1) to[bend right=10] (u1);
    \draw[placementedge] (v2) to[bend left = 10] (u2);
    \draw[placementedge] (v3) to[bend right= 10] (u5);
    \draw[placementedge] (v4) to[bend right = 5] (u4);

    \draw[routingedge, blue] (u2) to[bend right=15] (u3);
    \draw[routingedge, blue] (u3) to[bend right=15] (u1);
    \draw[routingedge, pink] (u4) to[bend left=15] (u5);
    \draw[routingedge, red] (u1) to[bend right=15] (u3);
    \draw[routingedge, red] (u3) to[bend right=15] (u5);
    \draw[routingedge, purple] (u4) to[bend right=15] (u2);
    \draw[routingedge, purple] (u2) to[bend right=15] (u1);
    \draw[routingedge, green] (u2) to[bend right=15] (u4);
        
    \end{tikzpicture}
    \caption{Example of a mapping of a virtual network (on the left) on a substrate network (on the right).}
    \label{fig:section1:example-vne}
\end{figure}

Figure \ref{fig:section1:example-vne} illustrates the embedding of a virtual graph on the left over a substrate graph on the right. The dotted arrows show the placement of virtual nodes. A virtual edge of a given color is routed using a substrate path of the same color. In this example, $m_V(\ubar_1) = u_1$, $m_V(\ubar_3) = u_5$ and $m_E(\ubar_1, \ubar_3) = \{(u_1, u_3), (u_3, u_5)\}$. \\

Using a unit of a substrate node (resp. edge) induces a utilization integer cost $w_u$ (resp. $w_e$). The cost of a mapping $m$, denoted $W_m$, corresponds to the sum of the placement and routing costs: $W_m = \sum_{\ubar \in V_r} d_{\ubar} w_{m_V(\ubar)} + \sum_{\ebar \in E_r} \sum_{e \in m_E(\ebar)} d_{\ebar} w_e$. The VNE can be formulated as follows:

\begin{definition}
    Given a virtual network $\Graph_r = (V_r, E_r)$ with demands $d$, a substrate network $\Graph_s = (V_s, E_s)$ with capacities $c$ and costs $w$, the Virtual Network Embedding Problem (VNE) is to find the minimum cost embedding of $\Graph_r$ on $\Graph_s$
\end{definition}

Introduced after the concept of network virtualization \cite{lu2006efficient, zhu2006algorithms}, the VNE is recognized as the central combinatorial problem behind resource allocation in network slicing \cite{vassilaras2017algorithmic}.

\subsection{Related works}

The Virtual Network Embedding problem is NP-complete \cite{amaldi2016computational, rost2020hardness}, even with uniform demands and topological restrictions \cite{benhamiche2025complexity}. A few special cases are known to be polynomial-time solvable, such as when the virtual network is a uniform-demand star \cite{benhamiche2025complexity, rost2015beyond}. 
Several VNE variants have been considered, such as path splitting for virtual edge routing, many-to-one node placement, directed graphs, or location constraints.  
A survey on many of these variants was conducted in 2013 \cite{fischer2013virtual}.
We now review the algorithms proposed to solve the VNE problem. \\

Most of the literature studies VNE in the online context, where small virtual networks arrive over time and must be embedded quickly, to maximize long-term revenue and acceptance rates \cite{chowdhury2011vineyard}. In that context, research has mostly focused on heuristics. 
Early works proposed greedy approaches \cite{lu2006efficient, yu2008rethinking}, later refined \cite{cheng2011virtual, gong2014toward, zhang2016virtual}, where first virtual nodes are assigned to substrate nodes based on ranking metrics (e.g., capacities, costs, node centrality, or clustering coefficient), and virtual edges are then routed using shortest-path or multi-commodity flow algorithms. 
To improve solutions obtained from those greedy heuristics, some authors have considered local search techniques \cite{infuhr2013solving, rguez2025grasp} and metaheuristics \cite{zhang2013unified,zhang2020vne,fajjari2011vne,zhang2019virtual,liu2016optimal}.
Another approach proposed in several works is to decompose the virtual network into smaller subgraphs, and compute mappings for each subgraph and then connect them to obtain a mapping of the whole virtual network \cite{lu2006efficient, zhu2016modified, song2019divide}. 
Literature has also considered AI-based methods, with the adaptation of Monte Carlo Tree Search \cite{haeri2017virtual}, Reinforcement Learning \cite{yao2018novel, dolati2019deepvine} or Graph Neural Network \cite{habibi2020accelerating, yan2020automatic} techniques for the VNE problem. \\

Fewer works have explored mathematical programming approaches.
The widely used integer linear formulation for the VNE is the Flow Formulation (FF), based on flow variables and flow conservation constraints \cite{chowdhury2009virtual, melo2013optimal}. 
In \cite{chowdhury2011vineyard}, it is used to devise the rounding heuristic ViNE, extensively used in VNE literature for comparison.
The formulation can be easily adapted to additional constraints, such as latency or location constraints \cite{melo2013optimal}.
However, it is known to present weak bounds, as discussed in \cite{moura2018branch, rost2019virtual}.
The Path Formulation (PF) yields tighter bounds by assigning directly virtual edges to substrate paths. Because the set of possible paths grows exponentially, its linear relaxation is commonly solved through Column Generation (CG) \cite{mijumbi2015path, moura2018branch}.
Building on this approach, \cite{mijumbi2015path} proposed a Price-and-Branch heuristic, while \cite{moura2018branch} introduced a full Branch-and-Price algorithm.
When embedding multiple virtual networks simultaneously, the Embedding Formulation (EF) has been proposed, where each variable represents the embedding of a whole virtual network.
Column generation is used to devise a Price-and-Branch heuristic in \cite{jarray2014decomposition} and a random rounding heuristic in \cite{rost2019parametrized}. \\

\subsection{Motivation and contributions}

Our work is motivated by two limitations in the current \vne literature in the context of Network Slicing.
First, VNE has almost exclusively been studied in the online setting, where requests arrive sequentially and must be embedded very quickly. 
However, many long-term slices, such as those supporting industrial applications or critical infrastructure, must remain in place for long periods of time. 
Hence, the allocation of resources in the core network should be carefully optimized.  
This motivates the study of offline VNE on core networks.
Second, the recent implementations of country-wide slices \cite{Siemens2025water, Tmobile2025priority} indicate that real slices may span large portions of the infrastructure.
However, the literature of VNE considers small virtual networks (at most 20 nodes): to the best of our knowledge, only one work has tackled large virtual networks \cite{song2019divide}.
For such large virtual networks, state-of-the-art heuristics, based on greedy placement strategies, have poor performances and, when capacities are tight, often fail to find any feasible solutions. \\

We aim to propose algorithms that tackle the embedding of large, realistic virtual networks. To achieve this, we propose a decomposition approach based on a partitioning of the virtual network. The main contributions of this paper are the following:

\begin{itemize}
    \item We highlight the need for offline VNE in the context of long-term network slices. We introduce large instances based on real core network topologies \cite{schneider2025vnelib}.
    \item We introduce the Virtual Partition Formulation, an integer linear formulation based on a tuned partition of the virtual network.
    \item We solve the linear relaxation of this formulation with a column generation algorithm, to obtain lower bounds.
    \item We enhance the pricing scheme 
    with a so-called \textit{subgraph-to-subgraph} strategy
    to propose an efficient Price-and-Branch heuristic.
\end{itemize}

The paper is organized as follows.
Section~2 introduces a new set of benchmark instances for VNE problems with large virtual networks.
Section~3 reviews the classical Flow Formulation and presents our new integer model.
Section~4 describes a column generation algorithm used to compute strong lower bounds, which are shown experimentally to outperform those obtained with the Flow Formulation.
Section~5 discusses refinements of the pricing scheme that produce columns suitable for constructing high-quality solutions. Based on these developments, we propose a Price-and-Branch heuristic that outperforms state-of-the-art algorithms.
Finally, Section~6 concludes the paper.

\section{Large and realistic instances}\label{sec:instances}

\begin{table*}
\centering
\begin{tabular*}{\textwidth}{@{\extracolsep\fill}llllllll@{}}\toprule
\multicolumn{4}{c}{\textbf{Virtual Networks}} & \multicolumn{4}{c}{\textbf{Substrate Networks}} \\ \cmidrule{1-4} \cmidrule{5-8}
\textbf{Name} & $|V_r|$ & $|E_r|$ & \textbf{Origin} & \textbf{Name} & $|V_s|$ & $|E_s|$ & \textbf{Origin} \\ \midrule
IRIS & 51 & 64 & TopologyZoo & TATA & 145 & 186 & TopologyZoo \\
zib54 & 54 & 80 & SDNlib & GtsCe & 149 & 193 & TopologyZoo \\
HiberniaGlobal & 55 & 81 & TopologyZoo & Cogent & 197 & 243 & TopologyZoo \\
UsSignal & 61 & 77 & TopologyZoo & & & & \\
ta2 & 65 & 108 & SDNlib & & & & \\
TW & 71 & 115 & TopologyZoo & & & & \\
Intellifiber & 73 & 95 & TopologyZoo & & & & \\
Uninett2010 & 74 & 101 & TopologyZoo & & & & \\
OTEGlobe & 83 & 90 & TopologyZoo & & & & \\
Interroute & 110 & 148 & TopologyZoo & & & & \\
\bottomrule
\end{tabular*}
\caption{List of virtual and substrate network instances}
\label{table:instances}
\end{table*}

Since network virtualization has not yet reached widespread commercial deployment, no practical benchmark of instances exists. This has led previous works to generate random virtual topologies and demands, as well as substrate topologies, capacities, and costs \cite{zhu2006algorithms, chowdhury2011vineyard, zhang2017virtual, gong2014toward, zhang2013unified, melo2013optimal}. This over-reliance on randomness makes the instances hard to reproduce and arguably unrealistic. Moreover, the difficulty of an instance, which depends on topologies, demands and capacities, is hard to control. Instead, we propose a new benchmark of reusable and realistic offline VNE instances. These benchmark instances are publicly available at \cite{schneider2025vnelib} for use in future studies. We consider large virtual networks, since in practice slices may span large portions of the physical networks (e.g. \cite{Siemens2025water, Tmobile2025priority}). The generation process, described in detail in \cite{schneider2025vnelib}, is the following. \\ 

\textbf{Topologies:} 
In the literature, the standard experimental configuration used relies on small (less than 20 nodes) virtual networks, randomly generated with the GT-ITM tool \cite{zegura1996model}. Substrate networks (between 100 and 200 nodes) are also generated randomly from GT-ITM. 
The networks generated by GT-ITM are whole networks (core and access), and not only core ones. Here, a core node corresponds to a high-capacity transit switch or Point-of-Presence (PoP) of the backbone network, whereas an access node represents an entry point to a local network. Instead, we will use real large-scale backbone topologies, as resource allocation decisions on these core networks have the largest impact on long-term investment costs. 
Those topologies are taken from the Internet Topology Zoo \cite{knight2011internet} and the Survivable Network Design Library (SNDlib) \cite{orlowski2010sndlib}, two widely-used collections in network optimization. These datasets provide Points-of-Presence level representations of backbones from real Internet Infrastructure Providers, at the scale of a country or continent. We selected the largest topologies available in each dataset, which are summarized in Table~\ref{table:instances}.  \\


\textbf{Demands:} 
Generating realistic demand patterns remains a challenging task, although some attempts have been made in this direction by considering traffic matrices \cite{lu2006efficient, da2020impact}.
This leads us to consider, for this study, the simple setting of assigning unit demands to all virtual nodes and edges. This configuration preserves the fundamental simultaneous placement and routing structure of VNE, while removing packing effects that are not central to its difficulty. The insights obtained for this setting remain representative of the general VNE problem. \\

\textbf{Capacities and costs:}
Since we consider unit demands and one-to-one node placement, substrate nodes have either a null or a unit capacity.
Substrate nodes and edges capacities are generated randomly, but according to three types: large, medium, and small, as follows. Large capacities are high enough that routing feasibility is never an issue, and any node placement will induce a solution. These instances allow us to evaluate that algorithms can find good solutions when feasibility is not an issue. For medium capacities, edge capacities are much more restrictive, such that edge routing becomes difficult. In small capacities, some substrate nodes (20\%) have null capacity, 
and can only be used as passing nodes for the edge routing paths.
The edge capacities are even tighter, such that embedding on these substrate networks should be very challenging.
Costs are defined following a common principle in real networks: resources with higher capacity and connectivity are cheaper to use, thereby encouraging load balancing.
The exact generation of capacities and costs is detailed in \cite{schneider2025vnelib}.
Note that capacities are generated in such a way that every feasible solution for a small-capacity instance is also feasible (with the same cost) in the corresponding medium instance, and likewise for medium vs. large. Hence, the optimal value for small (resp. medium) capacities is always greater than or equal to that of the corresponding medium (resp. large) instance.

\section{A new formulation based on virtual graph partitioning}
\label{sec:partitioning_formulation}

In this section, we present the Flow Formulation used for VNE. We discuss the weaknesses of this formulation, leading us to introduce a new formulation based on a partition of the virtual network.
In the following, we denote $v^*$ the optimal value of a given instance. For an integer linear formulation $(F)$, its linear relaxation, obtained by relaxing the integrality constraints, is denoted as $(\widetilde{F})$, and the optimal value of this relaxation as $v(\widetilde{F})$.

\subsection{The (undirected) Flow Formulation}

In what follows, we describe the Flow Formulation. It was introduced for the directed case in \cite{chowdhury2009virtual, melo2013optimal}, 
and adapted for the undirected case in \cite{benhamiche2026polytope}, for the one-to-one node placement case. Although the edges are undirected, to apply flow conservation constraints, we consider that the edges of $E_r$ and $E_s$ are an ordered set, which is equivalent to giving them an arbitrary direction.
A virtual edge can traverse a substrate edge $(u,v) \in E_s$ in both direction $(u,v)$ or $(v,u)$, since the network is undirected.
To capture this, we introduce bidirected flow variables: for each substrate edge, we define two variables, one per orientation (original and reverse). 
This technique is standard in flow problems \cite{ahuja1988network} or network design problems \cite{raack2011cut}. \\

Two sets of binary variables are considered. The placement variable $x_{\ubar u}$ takes the value 1 if and only if the virtual node $\ubar \in V_r$ is placed on substrate node $u \in V_s$.
The flow variable $y_{\ebar e_+}$ (resp. $y_{\ebar e_-}$) takes the value 1 if and only if the path routing virtual edge $\ebar \in E_r$ uses substrate edge $(u,v) = e \in E_s$, from $u$ to $v$ (resp. from $v$ to $u$). 
Note that a straightforward pre-treatment is to set to 0 the variable $x_{\ubar u}$ (resp. $y_{\ebar e}$) if $d_{\ubar} > c_{u}$ (resp. $d_{\ebar e} > c_{e}$), for any $\ubar \in V_r, u \in V_s$ (resp. any $\ebar \in E_r, e examp\in E_s$).
With this, ensuring one-to-one node placement is enough to ensure that node capacities are respected. In the many-to-one case, however, a capacity constraint would be necessary instead.
For an arc $e$, its start (resp. terminal) node is denoted $s(e)$ (resp. $t(e)$). 
For a node $u \in V_s$, its adjacent edges are denoted $\delta(u)$. Additionally, we denote $\delta^+(u) = \{e_+:s(e)=u\} \cup \{ e_-: t(e)=u\}$, which corresponds to the outgoing arcs of the node $u$ on the bidirected version of the substrate network. 
The Flow Formulation $(FF)$ is the following.

\begin{subequations}
\begin{align}
    \min \quad & \sum_{u \in V_s} \sum_{\ubar \in V_r} d_{\ubar} w_u  x_{\ubar u} \notag \\
    & \quad + \sum_{e \in E_s} \sum_{\ebar \in E_r} d_{\ebar} w_e (y_{\ebar e_+} + y_{\ebar e_-}) \notag \\
    \text{s.t.} \quad 
    & \sum_{u \in V_s} x_{\ubar u} = 1 && \forall \ubar \in V_r  \label{ff1}\\
    & x_{s(\ebar) u} - x_{t(\ebar) u} = \sum_{e \in \delta(u)} (y_{\ebar e_+} - y_{\ebar e_-}) && \forall \ebar \in E_r,\ \forall u \in V_s \label{ff2} \\
    & \sum_{\ubar \in V_r} x_{\ubar u} \leq 1 && \forall u \in V_s \label{ff1t1} \\
    & \sum_{\ebar \in E_r} d_{\ebar} (y_{\ebar e_+} + y_{\ebar e_-}) \leq c_e && \forall e \in E_s \label{ff4} \\
    & \sum_{e : s(e) = u} y_{\ebar e} \geq x_{s(\ebar) u} && \forall \ebar \in E_r, \forall u \in V_s \label{ff5} \\
    & x_{\ubar u} \in \{0,1\} && \forall \ubar \in V_r, \forall u \in V_s,  \nonumber \\
    & y_{\ebar e_+},  y_{\ebar e_-} \in \{0,1\} && \forall \ebar \in E_r, \forall e \in E_s \nonumber
\end{align}
\end{subequations}

The objective function expresses the embedding cost, taking into account placement and routing costs.
The constraints (\ref{ff1}) ensure that each virtual node is placed on exactly one substrate node. The constraints (\ref{ff2}) are the flow conservation constraints and guarantee that each virtual edge is associated with a routing path in the substrate graph whose end nodes host its endpoints. 
The one-to-one node placement is enforced by constraints (\ref{ff1t1}), which also ensure that node capacities are respected.
Constraints (\ref{ff4}) ensure that edge capacities are respected.  
Finally, the Constraints (\ref{ff5}) are the flow departure constraints, 
which are valid inequalities introduced in \cite{benhamiche2026polytope} for VNE with one-to-one node placement.
They permit to greatly reinforce the linear relaxation. They translate that, when a virtual node is placed on a substrate node, adjacent edges to the virtual node are routed on adjacent edges to this substrate node.  \\

\subsection{The new Virtual Partition Formulation}

Unfortunately, the Flow Formulation suffers from poor linear relaxation \cite{rost2019virtual, moura2018branch}, which we have observed in our preliminary tests, even with the flow departure constraints. To obtain better lower bounds, the Path Formulation has been proposed \cite{mijumbi2015path, moura2018branch}, where variables indicate whether a virtual edge is routed using a given substrate path. 
Since there is an exponential number of variables, the linear relaxation is solved by column generation (CG) \cite{desrosiers2005primer}. However, in our early tests, the Path Formulation only proposed a slightly better linear relaxation than the Flow Formulation reinforced with flow departure constraints, at the cost of longer running times. \\

Our idea is that, to get a better linear relaxation, a decomposition should consider a larger sub-component than a single virtual edge: a \textit{virtual subgraph} of several connected nodes and edges. 
This is achieved through a partition of the virtual network into smaller subgraphs, which can be embedded efficiently through existing methods. 
This idea has practical relevance: it can be expected that the virtual network is composed of cohesive regions, e.g., part of a slice can be dedicated to a specific service or location, which can be tackled with dedicated routines.
Solving VNE through a partition of the virtual network into stars \cite{lu2006efficient} or cycles and paths \cite{zhu2016modified} has already been studied in a few heuristics \cite{song2019divide, lu2006efficient, zhu2016modified}. 
In this work, we propose a formal framework, based on a new integer linear formulation for \vne: the \vpflong.  \\

Given a subset of virtual nodes $V'_r \subseteq V_r$, the \textit{induced subgraph} of $\Graph_r$ on $V'_r$ is the graph $\Graph'_r = (V'_r, E'_r)$, where $E'_r = \{ (u, v) \in E_r : u, v \in V'_r\}$. In the following, we will often omit the word \textit{induced}.
Now consider a partition of the virtual nodes of $\Graph_r$ into $1 \leq k_r \leq n_r$ subsets $\pi = (V_r^1, \ldots, V_r^{k_r})$, such that $V_r = \cup_{i \in \{1, \ldots, k_r \}} V^i_r$ (every node belongs to exactly one subset) and $|V^i_r| \geq 1$ for $i \in \{1, \ldots, k_r\}$.
For each part $V_r^i$, we denote the corresponding induced subgraph by $\GraphH_r^i$. The subgraph containing a node $\ubar \in V_r$ is denoted $\GraphH_r(\ubar)$.
Edges that connect nodes in different subsets do not belong to any subgraph $\GraphH_r^i$. We denote this set of \textit{cut edges} by $E_r^0$. 
The set of feasible sub-mappings of a subgraph $\GraphH_r^i$ is denoted $\Mapping_i$. 
The cost of a sub-mapping $m \in \Mapping_i$ is denoted $w_m$. $\textsc{x}^m_{\ubar u}$ is equal to one iff $\ubar$ is placed on $u$ in $m$, $\textsc{y}^m_{\ebar e_+}$ (resp. $\textsc{y}^m_{\ebar e_-}$) is equal to one iff the path routing $\ebar$ uses substrate edge $(u,v) = e \in E_s$, from $u$ to $v$ (resp. from $v$ to $u$) in $m$.  \\

Two sets of binary variables are considered. The \textit{sub-mapping} variables, denoted $\lambda^i_m$, take the value 1 if and only if the mapping $m \in \Mapping_i$ is selected for subgraph $\GraphH^i_r$, $i \in \{1, \ldots, k_r\}$.
Flow variable $y_{\ebar e_+}$ (resp. $y_{\ebar e_-}$) takes the value 1 if and only if the path routing virtual cut edge $\ebar \in E_r^0$ uses substrate edge $(u,v) = e \in E_s$, from $u$ to $v$ (resp. from $v$ to $u$).
The Virtual Partition Formulation for the partitioning $\pi$, denoted $(VPF_\pi)$, is the following. \\

\begin{figure*}[t]
\begin{subequations}
    \begin{align}
    \min \quad	& \sum_{i \in \{1, \ldots, k_r\}} \sum_{m \in \Mapping_i}  w_m \lambda_m^i + \sum_{e \in E_s} \sum_{\ebar \in E_r^0} w_e d_{\ebar} (y_{\ebar e_-} + y_{\ebar e_+}) \nonumber \\
    \text{s.t.}	\quad & \sum_{m \in \Mapping_i} \lambda_m^i \geq 1 && \forall i \in \{1, \ldots, k_r\} \label{VPF1} \\
    & \sum_{ \substack{m \in \Mapping_i, \\ s(\ebar) \in V_r^i} } \textsc{x}^m_{s(\ebar) u} \lambda_m^i  - \sum_{ \substack{m \in \Mapping_i, \\ t(\ebar) \in V_r^i}} \textsc{x}^m_{t(\ebar)u} \lambda_m^i =  \sum_{e \in \delta(u)} (y_{\ebar e_+} - y_{\ebar e_-}) && \forall \ebar \in E_r^0, \forall u \in V_s \label{VPF2}  \\
    & \sum_{i \in \{1, \ldots, k_r\}} \sum_{m \in \Mapping_i} \sum_{\ubar \in V_r^i}   \textsc{x}^m_{\ubar u} \lambda_m^i  \leq 1 && \forall u \in V_s  \label{VPF1t1} \\
    &  \sum_{i \in \{1, \ldots, k_r\}} \sum_{m \in \Mapping_i} \sum_{\ebar \in E^i_r} d_{\ebar} (\textsc{y}^m_{\ebar e_-} +\textsc{y}^m_{\ebar e_+} ) \lambda_m^i  + \sum_{\ebar \in E_r^0}  d_{\ebar} (y_{\ebar e_+} + y_{\ebar e_-})  \leq c_e	&& \forall e \in E_s  \label{VPF4} \\
    & \sum_{\substack{m \in \Mapping_i, \\ s(\ebar) \in V_r^i}} \textsc{x}^m_{s(\ebar) u} \lambda^i_m \le \sum_{e \in \delta^+(u)} y_{\ebar e} && \forall \ebar \in E_r^0, \forall u \in V_s   \label{VPF5}  \\
    & \lambda_m^i \in \{0,1\} && \forall m \in \Mapping_i,\, \forall i \in \{1, \ldots, k_r\} \nonumber \\
    & y_{\ebar e_+},  y_{\ebar e_-} \in \{0,1\} && \forall \ebar \in E_r, \forall e \in E_s \nonumber
\end{align}
\end{subequations}
\end{figure*}

Constraints (\ref{VPF1}) are the \textit{convexity constraints} of the extended formulation, ensuring that one submapping is selected for each subgraph. Constraints (\ref{VPF2}) (resp. (\ref{VPF5})) are the flow conservation (resp. flow departure) constraints for the virtual cut edges. They are adapted from their counterpart in $(FF)$. 
The one-to-one node placement is enforced by constraints (\ref{VPF1t1}), which also ensure that node capacities are respected. 
Constraints (\ref{VPF4}) ensure that substrate edge capacities are respected.  
Note that this formulation can be easily adapted to the many-to-one node placement case, by rewriting constraints (\ref{VPF1t1}) as node capacity constraints.

\subsection{Formulations comparison}

This decomposition is quite unusual, as columns have the same structure (a feasible mapping for a virtual network) as the solution for the overall problem. Such a decomposition is similar to what is done in \cite{caprara2016solving}, on a Temporal Knapsack problem, where it is coined as a \textit{Recursive Dantzig-Wolfe Decomposition}.
Moreover, $(VPF)$ has a similar structure to $(FF)$: the flow conservation and flow departure constraints remain for the virtual cut edges, whereas placement of nodes is replaced by the mapping of a subgraph of the virtual network. 
If a subgraph corresponds to a single virtual node, then the sub-mapping variable corresponds to a classical placement variable from the flow formulation. In that case, $(VPF)$ and $(FF)$ are the same formulation. On the other hand, if $k_r = 1$, then the only subgraph is the entire virtual network and $v(\vpfrelax) = v^*$. In what follows, we will study a few properties of $(VPF)$, using the instance from Example~\ref{ex:vpf-props} for illustration.

\begin{myexample} \label{ex:vpf-props}
    Figure~\ref{fig:vpf-vs-ff} provides an instance of VNE. The virtual network, on the left, has unit demands on all nodes and edges. The substrate network, on the right, has unit capacities on all nodes, 
    infinite capacity on edges,
    with zero cost on nodes and unit cost on edges. For this instance,  $v(\ffrelax) = |E_r|$, 
    with the following fractional point:

\noindent $x_{\ubar_1 u_3} = x_{\ubar_1 u_4} = x_{\ubar_1 u_7} = x_{\ubar_1 u_8} = 0.25$, \\
$x_{\ubar_2 u_3} = x_{\ubar_2 u_4} = x_{\ubar_2 u_5} = x_{\ubar_2 u_6} = 0.25$, \\
$x_{\ubar_3 u_2} = x_{\ubar_3 u_3} = x_{\ubar_3 u_6} = x_{\ubar_3 u_7} = 0.25$, \\
$x_{\ubar_4 u_1} = x_{\ubar_4 u_4} = x_{\ubar_4 u_5} = x_{\ubar_4 u_6} = 0.25$, \\
$x_{\ubar_5 u_3} = x_{\ubar_5 u_4} = x_{\ubar_5 u_6} = x_{\ubar_5 u_7} = 0.25$, \\
$x_{\ubar_6 u_1} = x_{\ubar_6 u_2} = x_{\ubar_6 u_5} = x_{\ubar_6 u_8} = 0.25$, \\
$y_{(\ubar_1, \ubar_2) (u_3, u_4)} = y_{(\ubar_1, \ubar_2) (u_4, u_3)} = y_{(\ubar_1, \ubar_2) (u_7, u_6)} =y_{(\ubar_1, \ubar_2) (u_8, u_5)} = 0.25,$ \\
$y_{(\ubar_2, \ubar_3) (u_3, u_2)} = y_{(\ubar_2, \ubar_3) (u_4, u_3)} = y_{(\ubar_2, \ubar_3) (u_5, u_6)} = y_{(\ubar_2, \ubar_3) (u_6, u_7)} = 0.25$ \\
$y_{(\ubar_1, \ubar_3) (u_3, u_2)} = y_{(\ubar_1, \ubar_3) (u_4, u_3)} = y_{(\ubar_1, \ubar_3) (u_7, u_6)} =y_{(\ubar_1, \ubar_3) (u_8, u_7)} = 0.25,$ \\
$y_{(\ubar_1, \ubar_4) (u_3, u_4)} = y_{(\ubar_1, \ubar_2) (u_4, u_1)} = y_{(\ubar_1, \ubar_2) (u_7, u_6)} =y_{(\ubar_1, \ubar_2) (u_8, u_5)} = 0.25,$ \\
$y_{(\ubar_4, \ubar_5) (u_1, u_4)} = y_{(\ubar_5, \ubar_6) (u_4, u_3)} = y_{(\ubar_5, \ubar_6) (u_5, u_6)} =y_{(\ubar_4, \ubar_5) (u_6, u_7)} = 0.25,$ \\ 
$y_{(\ubar_4, \ubar_6) (u_1, u_2)} = y_{(\ubar_4, \ubar_6) (u_4, u_1)} = y_{(\ubar_4, \ubar_6) (u_5, u_8)} =y_{(\ubar_4, \ubar_6) (u_6, u_5)} = 0.25,$ \\
$y_{(\ubar_5, \ubar_6) (u_3, u_2)} = y_{(\ubar_5, \ubar_6) (u_4, u_1)} = y_{(\ubar_5, \ubar_6) (u_6, 5)} =y_{(\ubar_5, \ubar_6) (u_7, u_8)} = 0.25,$ \\
and all other variables equal 0. 
    
    \begin{figure}[h]
    \centering
    \begin{tikzpicture}
    \begin{scope}[xshift=-2.25cm]
        \node at (0,-3.5) {$\Graph_r$};
        \node[vnode] (v1) at (0,0.65){$\ubar_1$};
        \node[vnode] (v2) at (-0.75,1.55){$\ubar_2$};
        \node[vnode] (v3) at (0.75,1.55){$\ubar_3$};
        \node[vnode] (v4) at (0,-0.65){$\ubar_4$};
        \node[vnode] (v5) at (-0.75,-1.55){$\ubar_5$};
        \node[vnode] (v6) at (0.75, -1.55){$\ubar_6$};
        
        \draw[networkedge] (v1) -- (v2);
        \draw[networkedge] (v2) -- (v3);
        \draw[networkedge] (v3) -- (v1);
        \draw[networkedge] (v1) -- (v4);
        \draw[networkedge] (v4) -- (v5);
        \draw[networkedge] (v5) -- (v6);
        \draw[networkedge] (v6) -- (v4);
    \end{scope}

    \begin{scope}[xshift=2.25cm]
        \node at (0,-3.5) {$\Graph_s$};
        \node[snode] (u1) at (0,0.75){$u_1$};
        \node[snode] (u2) at (-1, 1.75){$u_2$};
        \node[snode] (u3) at (0,2.75){$u_3$};
        \node[snode] (u4) at (1, 1.75){$u_4$};
        \node[snode] (u5) at (0,-0.75){$u_5$};
        \node[snode] (u6) at (-1,-1.75){$u_6$};
        \node[snode] (u7) at (0,-2.75){$u_7$};
        \node[snode] (u8) at (1,-1.75){$u_8$};
        
        \draw[networkedge] (u1) -- (u2);
        \draw[networkedge] (u2) -- (u3);
        \draw[networkedge] (u3) -- (u4);
        \draw[networkedge] (u4) -- (u1);
        \draw[networkedge] (u5) -- (u6);
        \draw[networkedge] (u6) -- (u7);
        \draw[networkedge] (u7) -- (u8);
        \draw[networkedge] (u8) -- (u5);
        \draw[networkedge] (u1) -- (u5);
    \end{scope}

    \end{tikzpicture}
    \caption{A VNE instance}
    \label{fig:vpf-vs-ff}
    \end{figure}
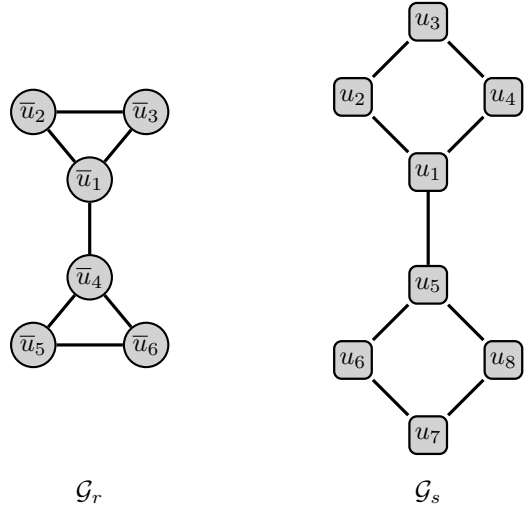
\end{myexample}

First, let us show that our new formulation has a better linear relaxation than the \fflong.

\begin{proposition}
    $v(\ffrelax) \leq v(\vpfrelax_\pi)$, for any partition $\pi$ of the virtual network, and the gap between the two can be arbitrarily large. 
\end{proposition}

\begin{proof}
    Consider a $k_r$ partition $\pi$ of the virtual network.
    Given a fractional solution $(\lambda, y)$ to $(\vpfrelax_\pi)$, a valid solution $(x, y')$ to $(\ffrelax)$ can be built as follows:
    
    \begin{itemize}
        \item For $i \in \{1, \ldots, k_r\}$, for $\ubar \in V_r^i$, set $x_{\ubar u} = \sum_{m \in \Mapping_i} \textsc{x}^m_{\ubar u} \lambda^i_m$
        \item For $i \in \{1, \ldots, k_r\}$, for $\ebar \in E_r^i$, set $y_{\ebar e_-} = \sum_{m \in \Mapping_i} \textsc{y}^m_{\ebar e_-} \lambda^i_m$ and $y_{\ebar e_+} = \sum_{m \in \Mapping_i} \textsc{y}^m_{\ebar e_+} \lambda^i_m$
        \item For $\ebar \in E_r^0$, set $y'_{\ebar e_-} = y_{\ebar e_-}$ and $y'_{\ebar e_+} = y_{\ebar e_+}$ for any $e \in E_s$.
    \end{itemize}
    
    Thus $v(\ffrelax) \leq v(\vpfrelax_{\pi})$. Now, consider the instance of Example~\ref{ex:vpf-props} and the partitioning of the virtual network in the two triangles $ \pi = \{ \{ \ubar_1, \ubar_2, \ubar_3\}, \{ \ubar_4, \ubar_5, \ubar_6 \} \}$
    Any sub-mapping of one of the two virtual triangle subgraphs has a cost of 4, since they have to be embedded on one of the substrate squares. Thus $v(\vpfrelax) = |E_s| > v(\ffrelax)$. 
    By considering larger substrate cycles instead of the two squares, the gap between the two relaxations can be increased to an arbitrarily large value.
\end{proof}

We will now prove a stronger result. A partition $\pi'$ is a \textit{refinement} of $\pi$ if every part of $\pi'$ is included in some part of $\pi$.

\begin{proposition}
    Let $\pi$ and $\pi'$ be two partitions of $V_r$ such that $\pi'$ is a refinement of $\pi$. Then $v(\vpfrelax_{\pi}) \ge v(\vpfrelax_{\pi'})$, and the gap between the two can be arbitrarily large. 
\end{proposition}

\begin{proof}
    We note $\pi = \{V_r^1, \ldots, V_r^{k_r}\}$ and $\pi' = \{{V'}_r^1, \ldots, {V'}_r^{k'_r}\}$, with $k_r \ge k'_r$.
    For $i \in \{1, \ldots k'_r\}$, let $\phi(i)$ denote the index of the part $V_r^{\phi(i)}$ containing ${V'}_r^i$.
    Given a solution $(\lambda, y)$ to $(\vpfrelax_{\pi})$, a solution $(\lambda', y')$ to $(\vpfrelax_{\pi'})$ of same costs can be built as follows.

   \begin{itemize}
        \item For $i \in \{1, \ldots, k'_r\}$, for each sub-mapping $m \in \Mapping_{\phi(i)}$ of $\GraphH_r^{\phi(i)}$, a sub-mapping $m'$ for $\GraphH'^i$ can be constructed by considering the node placement and edge routing of nodes and edges of $\GraphH^i$ in $m$. 
        Set ${\lambda'}^i_{m'} = \lambda^{\phi(i)}_m$.
        \item For $\ebar \in E_r^0$, simply set $y'_{\ebar e_-} = y_{\ebar e_-}$ and $y'_{\ebar e_+} = y_{\ebar e_+}$.
   \end{itemize}
    
    Thus $v(\vpfrelax_{\pi}) \ge v(\vpfrelax_{\pi'})$.
    Now, consider again the instance of Example~\ref{ex:vpf-props}. With $\pi = \{ \{ \ubar_1, \ubar_2, \ubar_3\}, \{ \ubar_4, \ubar_5, \ubar_6 \} \}$ and $\pi' = \{ \{\ubar_1\}, \{\ubar_2\}, \{\ubar_3\}, \{\ubar_4\}, \{\ubar_5\}, \{\ubar_6\} \}$, $\pi'$ is a refinement of $\pi$. As shown earlier $v(\vpfrelax_{\pi}) = |E_s| > v(\ffrelax) = v(\vpfrelax_{\pi'}) = |E_r|$, and the gap can be increased to an arbitrarily large value with larger substrate cycles.
\end{proof}

In this example, using a smaller $k_r$ leads to better relaxation values. We observe this behavior in practice, as shown in the next section. However, we can find theoretical instances where it is not the case:

\begin{remark}
    Let $\pi$ and $\pi'$ be two partitions of $V_r$ of sizes respectively $k_r$ and $k'_r$, with $k_r > k'_r$.
    In some cases, $v(\vpfrelax_{\pi}) > v(\vpfrelax_{\pi'})$.
\end{remark}

\begin{proof}
    Consider again the instance of Example~\ref{ex:vpf-props} and the partitions $\pi = \{ \{ \ubar_1, \ubar_2, \ubar_3\}, \{\ubar_4\}, \{\ubar_5\}, \{\ubar_6\}  \}$ and $\pi' = \{ \{ \ubar_1, \ubar_2 \} \{ \ubar_3, \ubar_4\}, \{\ubar_5, \ubar_6\} \}$. Although $\pi$ has more parts, $v(\vpfrelax_{\pi}) > v(\vpfrelax_{\pi'})$.
\end{proof}

In this example, a better bound is obtained with a denser (cyclic) subgraph, rather than larger (in average) path subgraphs. 
Finally, note that these results can be extended to the many-to-one node placement case.

\section{Column generation and lower bounds}

In this section, we describe our column generation implementation to solve the linear relaxation $(\vpfrelax)$ and compute strong lower bounds.

\subsection{Automatic partition of the virtual network}

First, we propose a partitioning strategy. The quality of the virtual network partition is crucial for the efficiency of our algorithm, as suggested by the observations in Section~3.3. We opt for two widely used automatic graph partitioners, KaHIP \cite{sanders2013think} and METIS \cite{karypis1997metis}.
Both algorithms aim to produce balanced partitions while minimizing the number of cut edges, which is particularly desirable in our case.
In practice, we found that KaHIP usually yields higher quality partitions, but the subgraphs produced are occasionally disconnected, and correcting the issue can significantly unbalance the partition. 
In such cases, we resort to METIS, which allows us to enforce connectivity explicitly. 
For the graph sizes considered in our experiments, both algorithms run extremely fast (below one millisecond).

\subsection{Column generation for $(\vpfrelax)$}

Column generation is an algorithm that solves linear programs with an exponential number of variables \cite{desrosiers2005primer}, by adding iteratively its variables, or \textit{columns}. 
The Restricted Master Problem $\rmp$ refers to $\vpf_{\pi}$ with a subset of columns. A \textit{Pricer Problem} generates the column with the lowest reduced costs. When no columns with negative reduced costs are found, $(\rmprelax)$ has been solved to optimality. In our case, for a given virtual subgraph $\GraphH_r^i$, $i\in \{1,\ldots,k\}$, the pricing problem $PP_i$ is actually a \vne problem of the virtual subgraph $\GraphH_r^i$ on $\Graph_s$. It can be solved using the Flow Formulation, with the following objective function:  

\begin{equation}
\begin{split}
    v(PP_i) = \min - \theta_{i} + & \sum_{u \in V_r^i} \sum_{\ubar \in V_s} (d_{\ubar} w_u - \phi_u) x_{\ubar u} \\
        & + \sum_{e \in E_s} \sum_{\ebar \in E_r^i} d_{\ebar} (w_e - \beta_e) y_{\ebar e} \\
        & +\sum_{\substack{\ebar \in E_r^0 \\ t(\ebar) \in V_r^i}} \sum_{u \in V_s} \alpha_{\ebar u} x_{t(\ebar) u} -  \sum_{\substack{\ebar \in E_r^0 \\ s(\ebar) \in V_r^i}} \sum_{u \in V_s} \alpha_{\ebar u} x_{s(\ebar) u}  \\ 
        & - \sum_{u \in V_s} \sum_{\substack{\ebar\in E_r^0, \\ s(\ebar) \in \GraphH^i_r}} \gamma_{\ebar u} x_{s(\ebar) u}
\end{split}
\end{equation}

\noindent where the dual variables $\theta, \alpha, \phi, \beta, \gamma$ correspond to respectively inequalities (\ref{VPF1}), (\ref{VPF2}), (\ref{VPF1t1}), (\ref{VPF4}), (\ref{VPF5}). \\

Although the convergence of the column generation algorithm can take a long time, the \textit{Lagrangian bound} proposes a valid lower bound for the problem and can be computed easily at each step of the algorithm:

\begin{equation}
    LGB = v(\rmprelax) + \sum_{i \in \{1, \ldots, k_r\}} v^*(PP_i)
\end{equation}

\subsection{Stabilization}

We observed oscillations in the dual variables during our tests, in particular for the dual variables $\beta$, related to the flow conservation constraints, which can take both positive and negative values. 
At a given iteration of the column generation, when dual variables are far from their optimal values, it results in weak columns. This behaviour is a well-known issue in column generation.
To mitigate these oscillations, we adopt a standard stabilization technique: dual variable smoothing \cite{pessoa2013out}.
At each iteration, instead of using the current dual variables from the restricted master problem, we use a convex combination of the previous stabilized duals and the current ones. The parameter $\alpha \in [0, 1]$ controls the smoothing strength.
At a given iteration of the column generation, the stabilized counterpart $\hat{\omega}$ of the dual variable $\omega$ is:
\begin{equation*}
    \hat{\omega} = \alpha \hat{\omega}' + (1 - \alpha) \omega
\end{equation*}
where $\hat{\omega}'$ is the stabilized variable at the previous iteration.
In practice, we use $\alpha = 0.9$, which showed good results in \cite{pessoa2013out}. 
Preliminary experiments showed that this stabilization improves the value of the Lagrangian bound obtained in one hour, for our large instances, by around 5\%.

\subsection{Sub-pricers}

Solving the pricing problems optimally can be very time-consuming. 
The most efficient algorithms to solve the pricing problems optimally are MILP solvers such as IBM ILOG CPLEX (referred to as CPLEX), using the Flow Formulation. However, they have long solving times even for quite small virtual networks on a large substrate network: often more than a minute with $n_r = 10$ and $n_s = 150$. 
Since exact pricers are too expensive, we will only use them at the end of the CG, in order to provide a strong lower bound. \\

Before, we use \textit{sub-pricers}, which are pricers restricted to a substrate subgraph. 
Note that such sub-pricers, even solved exactly, are still heuristics and cannot be used to produce a Lagrangian bound.
These substrate subgraphs of the sub-pricers are generated as follows. The substrate is partitioned into a large number of parts, using KaHIP or METIS as previously mentioned. Then each substrate subgraph is expanded by adding neighbouring nodes, up to roughly three times the virtual subgraph average size.
This method offers three advantages.
\begin{itemize}
    \item A sub-pricer contains very good columns, if the substrate subgraphs are large enough. Indeed, good sub-mappings will not usually span a large part of the substrate network, as it would induce high routing costs.
    \item Selecting the size of the subgraph provides a good way to balance column quality and sub-pricer difficulty (and runtime when using an exact solver).
    \item All constraints of the $(RMP)$ are covered, which is typically very helpful for the convergence of the CG.
    Indeed,  the columns are diverse and pave all parts of the substrate network, since we generated diverse substrate subgraphs.
\end{itemize}

\begin{algorithm}[b!]
\caption{Greedy}\label{alg:greedy}
\begin{algorithmic}[1]
\Require VNE instance
\Ensure a mapping for the instance (if one was found)
\State Initialize arrays $m_V[1..n_r] = [0..0]$
\State $m_V[$random$(1..n_r)] = $random$(1..n_s)$
\For{i in $1:n_r-1$}
  \State $\ubar \gets $random$(\vbar: m_V(\vbar)=0 : $
  \State $\qquad \qquad \qquad |\wbar \in N(\vbar):m_V(\wbar) \neq 0| \ge 1)$
  \State $l_u = \sum_{\vbar \in N(\ubar)} $ShortestPath$(m_V(\vbar), u)$
  \State $m_V[\ubar] = $argmin$(l_u)$
\EndFor
\State $m_E = $SuccessiveShortestPathRouting$(m_V)$
\State \Return $m=(m_V, m_E)$
\end{algorithmic}
\end{algorithm}

\begin{table*}
\centering
\begin{tabular*}{\textwidth}{@{\extracolsep\fill}lllllllllll@{}}\toprule 
  & \multicolumn{2}{c}{$\lfloor \frac{n_r}{10} \rfloor + 2$} & \multicolumn{2}{c}{$\lfloor \frac{n_r}{10} \rfloor + 1$} & \multicolumn{2}{c}{$\lfloor \frac{n_r}{10} \rfloor$} & \multicolumn{2}{c}{$\lfloor \frac{n_r}{10} \rfloor - 1$} & \multicolumn{2}{c}{$\lfloor \frac{n_r}{10} \rfloor - 2$} \\\cmidrule{2-3} \cmidrule{4-5} \cmidrule{6-7} \cmidrule{8-9} \cmidrule{10-11}
   \textbf{Virtual topology} & Large & Small & Large & Small & Large & Small & Large & Small   & Large & Small \\ \midrule
    IRIS & 180.6 & 207.0 & 180.8 & 206.4 & 181.7 & 212.4 & 183.5 & 214.8 & 184.1 & 211.1 \\
    zib54 & 217.2 & 260.2 & 216.8 & 256.2 & 212.5 & 260.6 & 213.5 & 237.8 & 0.0 & 224.3 \\
    HiberniaGlobal& 222.4 & 261.4 & 229.3 & 270.5 & 223.5 & 267.1 & 210.3 & 238.8 & 0.0 & 0.0 \\
    US & 225.2 & 256.3 & 227.2 & 259.4 & 230.8 & 273.0 & 213.4 & 252.9 & 202.9 & 254.8 \\
    ta2 & 273.3 & 329.0 & 282.0 & 345.9 & 276.8 & 355.0 & 269.5 & 328.9 & 0.0 & 0.0 \\
    TW & 297.4 & 371.8 & 293.3 & 367.1 & 286.4 & 381.4 & 287.2 & 382.7 & 266.5 & 358.7 \\
    Intellifiber & 271.4 & 311.4 & 275.1 & 321.0 & 285.3 & 331.0 & 262.7 & 312.1 & 256.0 & 301.0 \\
    Uninett2010& 289.1 & 333.9 & 287.2 & 344.4 & 268.6 & 321.4 & 274.2 & 337.8 & 270.4 & 325.5 \\
    OTEGlobe & 306.3 & 352.3 & 312.5 & 366.6 & 308.0 & 359.6 & 308.1 & 371.4 & 303.8 & 344.5 \\
    Interroute & 433.8 & 502.8 & 433.7 & 524.0 & 437.9 & 529.3 & 440.4 & 536.4 & 409.4 & 524.0 \\
    \textbf{Average} & 271.7 & 318.6 & \textbf{273.8} & 326.2 & 271.2 & \textbf{329.1} & 266.3 & 321.4 & 189.3 & 254.4 \\
    \textbf{Average gap to v(RMP)} & 0.14 & 0.10 & 0.15 & 0.11 & 0.20 & 0.12 & 0.29 & 0.23 & - & - \\
    \bottomrule
\end{tabular*}
\caption{Lower bound value and gap to the RMP value, depending on the number of virtual subgraphs $k_r$ (first row)}
\label{table:lb-kr}
\end{table*}

\noindent \textbf{Exact sub-pricers} correspond to solving a sub-pricer with an exact solver (CPLEX in our case). \\

\noindent \textbf{Greedy sub-pricers} Since solving sub-pricers exactly can still consume too much time and we want to generate quickly many columns at the first iterations, we also propose to use the greedy heuristic described in Algorithm~\ref{alg:greedy} to solve the sub-pricer. This greedy heuristic works as follows.
First, the virtual node placement is initialized, with a random virtual node being placed on a random substrate node. 
Then, virtual nodes are placed one by one. The next virtual node to be placed must have at least one neighbor already placed (the neighbors of a node $\ubar$ are denoted $N(\ubar)$). Substrate nodes are ranked according to their distance to other substrate nodes hosting neighbors of the virtual node (list $l$). The virtual node is placed on the best substrate node.  
When all virtual nodes are placed, the virtual edge routing is achieved by solving successive shortest path problems on residual edge capacities. 
This greedy heuristic is similar to those of works \cite{cheng2011virtual, gong2014toward, zhang2016virtual}. However, in these works, the node placement is achieved by ranking substrate nodes based on their capacity, cost, and connectivity. We found that ranking based on the distances was much more effective in our case.
When the virtual subgraph is small, the heuristic usually provides a reasonably good mapping, while running very fast (well under 0.01 seconds) because the substrate subgraph is small. Since it is so cheap, we run the heuristic 100 times and keep the mapping with the best reduced cost. Columns generated that way are near optimal, and greatly help the beginning of the CG convergence at a much lower computational cost than exact sub-pricers.

\begin{table*}
\centering
\begin{tabular*}{\textwidth}{@{\extracolsep\fill}llllllll@{}}\toprule
  \multicolumn{2}{c}{\textbf{Topology}} & \multicolumn{2}{c}{\textbf{Large}} & \multicolumn{2}{c}{\textbf{Medium}} & \multicolumn{2}{c}{\textbf{Small}} \\ \cmidrule{1-2} \cmidrule{3-4} \cmidrule{5-6} \cmidrule{7-8}
  \textbf{Virtual} & \textbf{Substrate} & \textbf{CG-LB} & \textbf{CPLEX} & \textbf{CG-LB} &  \textbf{CPLEX} & \textbf{CG-LB} &  \textbf{CPLEX} \\ \midrule
  IRIS & TATA & \textbf{195} & 169 & \textbf{191} & 162 & \textbf{223} & 197 \\
   & GtsCe & \textbf{193} & 176 & \textbf{194} & 177 & \textbf{220} & 196 \\
   & Cogent & \textbf{182} & 170 & \textbf{183} & 170 & \textbf{213} & 189 \\
  zib54 & TATA & \textbf{225} & 193 & \textbf{225} & 188 & \textbf{280} & 264 \\
   & GtsCe & \textbf{220} & 195 & \textbf{231} & 199 & \textbf{266} & 239 \\
   & Cogent & \textbf{213} & 186 & \textbf{217} & 184 & \textbf{261} & 216 \\
  HiberniaGlobal& TATA & \textbf{235} & 187 & \textbf{233} & 188 & \textbf{280} & 222 \\
   & GtsCe & \textbf{230} & 198 & \textbf{230} & 199 & \textbf{282} & 216 \\
   & Cogent & \textbf{224} & 187 & \textbf{225} & 186 & \textbf{268} & 210 \\
  US & TATA & \textbf{243} & 212 & \textbf{242} & 213 & \textbf{288} & 250 \\
   & GtsCe & \textbf{248} & 214 & \textbf{244} & 211 & \textbf{285} & 246 \\
   & Cogent & \textbf{231} & 200 & \textbf{231} & 199 & \textbf{273} & 226 \\
  ta2 & TATA & \textbf{299} & 245 & \textbf{311} & 251 & \textbf{384} & 307 \\
   & GtsCe & \textbf{304} & 256 & \textbf{311} & 260 & \textbf{377} & 299 \\
   & Cogent & \textbf{277} & 238 & \textbf{278} & 242 & \textbf{356} & 279 \\
  TW & TATA & \textbf{311} & 268 & \textbf{323} & 278 & \textbf{415} & 365 \\
   & GtsCe & \textbf{327} & 280 & \textbf{336} & 293 & 405 & \textbf{418} \\
   & Cogent & \textbf{287} & 259 & \textbf{300} & 265 & \textbf{382} & 347 \\
  Intellifiber & TATA & \textbf{299} & 254 & \textbf{299} & 258 & \textbf{359} & 308 \\
   & GtsCe & \textbf{303} & 264 & \textbf{306} & 264 & \textbf{353} & 294 \\
   & Cogent & \textbf{286} & 245 & \textbf{282} & 245 & \textbf{331} & 279 \\
  Uninett2010& TATA & \textbf{292} & 262 & \textbf{293} & 262 & \textbf{360} & 322 \\
   & GtsCe & \textbf{298} & 273 & \textbf{297} & 274 & \textbf{346} & 307 \\
   & Cogent & \textbf{269} & 253 & \textbf{269} & 253 & \textbf{322} & 288 \\
  OTEGlobe & TATA & \textbf{330} & 292 & \textbf{327} & 291 & \textbf{395} & 353 \\
   & GtsCe & \textbf{335} & 298 & \textbf{334} & 299 & \textbf{396} & 347 \\
   & Cogent & \textbf{308} & 276 & \textbf{305} & 278 & \textbf{360} & 314 \\
  Interroute & TATA & \textbf{492} & 448 & \textbf{490} & 449 & \textbf{607} & 568 \\
   & GtsCe & \textbf{501} & 455 & \textbf{497} & 456 & Inf. & \textbf{Inf.} \\
   & Cogent & \textbf{438} & 406 & \textbf{439} & 408 & \textbf{530} & 476 \\ 
   \multicolumn{2}{l}{Average} &\textbf{286.1} & 251.5 & \textbf{287.7} & 252.9 & \textbf{338.0} &  293.9 \\
   \multicolumn{2}{l}{\# best} & \textbf{30} &  0 & \textbf{30} & 0 & \textbf{28} & 2 \\
   \bottomrule
\end{tabular*}
\caption{Lower bounds obtained 
by CG-LB and CPLEX, on all instances with large, medium, and small capacities}
\label{tab:lowerbound}
\end{table*}

\subsection{Lower bounds results}\label{sec:lowerbounds}

In what follows, we study the performance of our column generation procedure. Our implementation is developed in Julia \cite{bezanson2017julia}, with Graphs.jl \cite{Graphs2021} and JuMP \cite{dunning2017jump}. We use the state-of-the-art solver IBM ILOG CPLEX 22.1. In preliminary tests, we tried other solvers such as Gurobi, but CPLEX was the best-performing one. All computational experiments were conducted on a machine equipped with an Intel\textsuperscript{\textregistered}Xeon\textsuperscript{\textregistered} E5-2650 v3 CPU running at 2.30 GHz.  \\

All the bounds were computed within one hour. Our algorithm is the following. First, 800 columns are generated through the greedy sub-pricers, which, on our setup, takes less than a minute. Then, exact sub-pricers are run for 20 minutes, or until 2500 columns have been generated. Finally, exact pricers are run, which permits the generation of lower Lagrangian bounds. In the tables, the value reported is this Lagrangian bound $LGB$. \\

First, we aim to assess the right partition size, i.e., the number of virtual subgraphs $k_r$. 
Experiments were run on the 10 virtual networks using Cogent as the substrate network, with large capacities.
The Lagrangian lower bounds obtained are reported in Table~\ref{table:lb-kr}.
We also report the average gap to the value of the $(\rmprelax)$, which provides a valuable indication of how far column generation is from convergence, and therefore, of the potential for further improvement of the lower bound $LGB$. 
In most instances, the algorithm still presents an important gap (more than 10\%) after one hour. 
The gaps are higher when the substrate network has large capacities and when the virtual networks are larger, which is expected since these instances have more feasible columns. 
Gaps are also higher for smaller $k_r$, i.e., larger virtual subgraphs, for which the sub-pricers become increasingly more difficult. However, smaller $k_r$ typically lead to better final values, as discussed in Section~3.3. Hence, the choice of $k_r$ permits a tradeoff between solving time and bound quality. In the results of our experiments, the best bound values are obtained for $k_r = \lfloor \frac{n_r}{10} \rfloor$ and $k_r = \lfloor \frac{n_r}{10} \rfloor +1$. \\ %

Next, Table~\ref{tab:lowerbound} shows the results of experiments over all instances. 
We report the Lagrangian bound we could produce after one hour with $k_r = \lfloor \frac{n_r}{10} \rfloor$, and compare it to the bound found by CPLEX, 
using the Flow Formulation,
after one hour.
They show that our column generation consistently finds clearly better lower bounds than CPLEX, beating it on 88 of the 90 instances.
Our method allows an average improvement over CPLEX of 13.9\% for large capacities, 13.8\% for medium capacities, and 11.9\% for small capacities. 
On small capacities, CPLEX performs slightly better. We attribute this to the fact that there are fewer $x$ variables to branch on, hence the branch-and-bound tree is further advanced after one hour.
Overall, these results show that the new Virtual Partition Formulation is much stronger than the Flow Formulation.

\begin{remark}
    The lower bounds of an instance are very close for both large and medium capacity settings, for both our method and CPLEX. This is despite the integer problem being sensibly more difficult, as shown in the next section. To understand this, we observed the fractional solutions, and observed that they spread over most parts of the substrate network. This allows those fractional solutions to keep a low utilization on most substrate edges, which will be only slightly affected by the medium capacity settings.
    On the other hand, small capacities present sensibly higher bounds, as both edge and node capacities are much more restrictive.
\end{remark}

\section{Price-and-Branch heuristic}

In this section, we will aim to find solutions for the large instances described in Section~\ref{sec:instances}.
Greedy heuristics, such as Algorithm~\ref{alg:greedy} or those of \cite{cheng2011virtual, gong2014toward, zhang2016virtual}, perform less well as the virtual network grows. We will show that the Virtual Partition Formulation can be used to find high-quality solutions for the instances considered in this paper.
A straightforward approach consists in solving $(RMP)$ with integer variables after the column generation phase. This is known as a \textit{Price-and-Branch heuristic} (or Restricted Master Problem heuristic). 
In many decomposition schemes, the columns generated are already well suited for integer solutions, and this heuristic performs well \cite{Wang2025PriceBranch}. 
However, in our case, this algorithm, with columns generated as in Section~4, yielded high-cost solutions, or no solutions at all.  
We identified issues in the columns produced using the procedure in Section~4. 
This leads us to propose a so-called \textit{subgraph-to-subgraph} adaptation of the Price\&Branch heuristic defined in the following. This method is flexible and significantly helps improve integer performance, as shown by the experimental studies.

\subsection{Intersecting columns}

We say that two columns \textit{intersect} when they share at least one substrate node for hosting virtual nodes. In that case, they can not be used together in an integer solution due to the node capacity constraints, although they can be used together in a linear solution of $(\rmprelax)$ with fractional values.
One issue we identified in the previously generated columns is that many of them intersected at the central nodes of the substrate network.
To illustrate this, consider Example~\ref{ex:heuristic} \\

\begin{myexample} \label{ex:heuristic}
    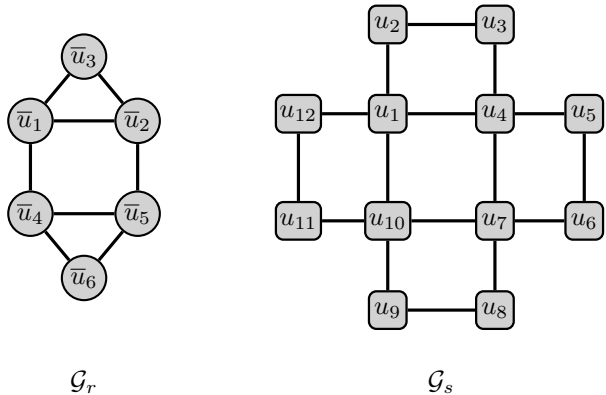
\begin{figure}[h]
\centering
\begin{tikzpicture}[scale=0.95] 
    \begin{scope}[xshift=-2.5cm]
        \node at (0,-3) {$\Graph_r$};
        \node[vnode] (v1) at (-0.75,0.65){$\ubar_1$};
        \node[vnode] (v3) at (0,1.55){$\ubar_3$};
        \node[vnode] (v2) at (0.75,0.65){$\ubar_2$};
        \node[vnode] (v4) at (-0.75,-0.65){$\ubar_4$};
        \node[vnode] (v6) at (0.,-1.55){$\ubar_6$};
        \node[vnode] (v5) at (0.75, -0.65){$\ubar_5$};
        
        \draw[networkedge] (v1) -- (v2);
        \draw[networkedge] (v2) -- (v3);
        \draw[networkedge] (v3) -- (v1);
        \draw[networkedge] (v1) -- (v4);
        \draw[networkedge] (v2) -- (v5);
        \draw[networkedge] (v4) -- (v5);
        \draw[networkedge] (v5) -- (v6);
        \draw[networkedge] (v6) -- (v4);
    \end{scope}

    \begin{scope}[xshift=2.5cm]
        \node at (0,-3) {$\Graph_s$};
        \node[snode] (u1) at (-0.75,0.75){$u_1$};
        \node[snode] (u2) at (0.75,0.75){$u_4$};
        \node[snode] (u3) at (0.75,-0.75){$u_{7}$};
        \node[snode] (u4) at (-0.75,-0.75){$u_{10}$};

        \node[snode] (u12) at (-2.,0.75){$u_{12}$};  
        \node[snode] (u13) at (-2.,-0.75){$u_{11}$};  
        
        \node[snode] (u22) at (-0.75,2){$u_{2}$};  
        \node[snode] (u23) at (0.75,2){$u_{3}$};  
        
        \node[snode] (u32) at (2,0.75){$u_{5}$};  
        \node[snode] (u33) at (2,-0.75){$u_{6}$};    
        
        \node[snode] (u42) at (0.75,-2){$u_{8}$};  
        \node[snode] (u43) at (-0.75,-2){$u_{9}$};

        \draw[networkedge] (u1) -- (u2);
        \draw[networkedge] (u2) -- (u3);
        \draw[networkedge] (u3) -- (u4);
        \draw[networkedge] (u4) -- (u1);
        
        \draw[networkedge] (u1) -- (u12);
        \draw[networkedge] (u12) -- (u13);
        \draw[networkedge] (u13) -- (u4);
        
        \draw[networkedge] (u1) -- (u22);
        \draw[networkedge] (u22) -- (u23);
        \draw[networkedge] (u23) -- (u2);
        
        \draw[networkedge] (u2) -- (u32);
        \draw[networkedge] (u32) -- (u33);
        \draw[networkedge] (u33) -- (u3);
        
        \draw[networkedge] (u3) -- (u42);
        \draw[networkedge] (u42) -- (u43);
        \draw[networkedge] (u43) -- (u4);
    \end{scope}

\end{tikzpicture}
\caption{A VNE instance}
\label{fig:overlapping-columns}
\end{figure}

Figure~\ref{fig:overlapping-columns} provides an instance of VNE. The virtual network (on the left) is two connected triangles, and has unit demands on every node and edge. The substrate network (on the right) is four connected squares, and has unit capacities on every node and edge, a null cost on every node, and a unit cost on every edge.
\end{myexample}

Take the partition $\pi = ( \{\ubar_1, \ubar_2, \ubar_3\}, \{\ubar_4, \ubar_5, \ubar_6\} )$ for this virtual network, which gives the two virtual subgraph $\GraphH_r^1$ and $\GraphH_r^2$. 
Consider the following two sub-mappings for $\GraphH_r^1$:

\begin{equation*}
\begin{aligned}
     m^{a}: \quad & m^{a}_V(\ubar_1) = u_1,\ m^{a}_V(\ubar_2) = u_{4},\ m^{a}_V(\ubar_3) = u_{2}, \\
     & m^{a}_E(\ubar_1, \ubar_2) = (u_1, u_{4}) \\
     & m^{a}_E(\ubar_2, \ubar_3) = (u_{4}, u_{3}, u_{2}), \\
     & m^{a}_E(\ubar_3, \ubar_1) = (u_{2}, u_1) \\
     m^{b}: \quad & m^{b}_V(\ubar_1) = u_{7},\ m^{b}_V(\ubar_2) = u_{10},\ m^{b}_V(\ubar_3) = u_{8}, \\
     & m^{b}_E(\ubar_1, \ubar_2) = (u_{7}, u_{10}), \\
     & m^{b}_E(\ubar_1, \ubar_2) = (u_{10}, u_{9}, u_{8}), \\
     & m^{b}_E(\ubar_3, \ubar_1) = (u_{8}, u_{7})
\end{aligned}
\end{equation*}
And, similarly, the two sub-mappings for $\GraphH^2_r$:
\begin{equation*}
\begin{aligned}
     m^{c}: \quad & m^{c}_V(\ubar_4) = u_{4},\ m^{c}_V(\ubar_5) = u_{7},\ m^{c}_V(\ubar_6) = u_{5}, \\
     & m^{c}_E(\ubar_4, \ubar_5) = (u_{4}, u_{7}) \\
     & m^{c}_E(\ubar_5, \ubar_6) = (u_{4}, u_{6}, u_{5}) \\
     & m^{c}_E(\ubar_6, \ubar_4) = (u_{5}, u_{4}) \\
     m^{d}:  \quad & m^{d}_V(\ubar_4) = u_{10},\ m^{d}_V(\ubar_5) = u_1,\ m^{d}_V(\ubar_6) = u_{11}, \\
     & m^{d}_E(\ubar_4, \ubar_5) = (u_{10}, u_1) \\
     & m^{d}_E(\ubar_5, \ubar_6) = (u_1, u_{12}, u_{11}) \\
     & m^{d}_E(\ubar_6, \ubar_4) = (u_{11}, u_{10})
\end{aligned}
\end{equation*}

The columns associated all have a cost of $4$. 
With these four columns, $v(\rmprelax) = 10$, with the solution:
\begin{equation*}
\begin{aligned}
     & \lambda^1_{m_{11}} = \lambda^1_{m_{12}} = 0.5, \\
     & \lambda^2_{m_{21}} = \lambda^2_{m_{22}} = 0.5,  \\
     & y_{(\ubar_1, \ubar_4)(u_1, u_{10})} = y_{(\ubar_1, \ubar_4)(u_{7}, u_{4})} = 0.5, \\
     & y_{(\ubar_2, \ubar_5)(u_{4}, u_1)} = y_{(\ubar_2, \ubar_5)(u_{10}, u_{4})} = 0.5
\end{aligned}
\end{equation*}
and the other $y$ variables are set to zero.
Here, $v(\rmprelax) = v(\vpfrelax_\pi) = v^*$.
However, no integer solutions can be constructed from these columns, as they intersect on nodes $u_1, u_{4}, u_{7}, u_{10}$. \\

In practice, we observed that the same central substrate nodes appeared in many columns, which resulted in poor performance of the Price-and-Branch heuristic.
To solve this issue, our key idea is to consider a 
virtual subgraph to substrate subgraph strategy, as follows.
Before the beginning of the column generation, the substrate network is partitioned, and, 
for each virtual subgraph,
columns are then generated with sub-pricers on these substrate subgraphs.
This 
{\it subgraph-to-subgraph}
approach has three advantages:

\begin{itemize}
    \item Columns generated on a substrate subgraph do not intersect with all columns from other substrate subgraphs.
    \item Assigning coherent parts of the virtual network to coherent parts of the substrate network will induce low-cost sub-mappings within the substrate network.
    \item Columns can be quickly generated when the substrate subgraphs are small (similar to Section~4.4).
\end{itemize}

In our experiments, this method significantly improves the performance of the Price-and-Branch heuristic. In many cases, only a few columns per virtual subgraph are sufficient to obtain feasible solutions. Adding more columns further improves solution quality, at the cost of increased computation time. This behavior is mainly due to the integer resolution of the restricted master problem $(RMP)$, whose complexity grows exponentially with the number of columns. In practice, solving the (RMP) with around 500 columns already often takes more than one minute. \\

With this implementation, the quality of the obtained solution depends strongly on the partitions computed for both the virtual and substrate networks. To generate these partitions, we use METIS and KaHIP, as described in Section~4.1. The size of the substrate subgraphs also plays a crucial role in the effectiveness of the heuristic. If the substrate subgraphs are too small, the sub-pricers may become infeasible or produce expensive columns. To illustrate this remark, consider a virtual and a substrate subgraph of the same number of nodes: when the virtual is denser than the substrate, 
the edge routing can quickly become unfeasible as many routing paths have to use the same substrate edges.
Conversely, very large subgraphs produce cheaper columns, but can 
force the sub-mapping to span a large part of the substrate subgraph, which can also result in long edge routing costs.
Based on our preliminary experiments, a good choice is to partition the substrate network into balanced subgraphs whose average size (number of capacitated nodes) is about 1.5 times the maximum size of the virtual subgraphs. This ratio, however, may vary depending on the demand, capacities, 
topologies, or if the many-to-one node placement case is considered.

\subsection{Better virtual boundary node placements}

In a graph partition, a \textit{boundary node} is a node adjacent to a cut edge, whereas a node with no neighbors outside its subgraph is called an \textit{internal node}.
A second issue we observed in the generated columns is that the virtual boundary nodes are often placed on internal substrate nodes. 
This induces longer routing paths for the virtual cut edge.
To illustrate this, consider again the instance of Example~\ref{ex:heuristic}. Take the partition $\pi = (\{\ubar_1, \ubar_2, \ubar_3\}, \{\ubar_4, \ubar_5, \ubar_6\} )$. 
Here, the virtual boundary nodes $\ubar_1, \ubar_2, \ubar_4, \ubar_5$ have to be placed on the substrate boundary nodes $u_1, u_{4}, u_{7}, u_{10}$, otherwise no solution exists.
However, throughout our experiments, we observed that virtual boundary nodes were often placed on outlier substrate nodes. \\

\begin{table*}[t]
\centering
\footnotesize                      
\setlength{\tabcolsep}{3pt}        
\begin{tabular*}{\textwidth}{@{\extracolsep\fill}lllllllllll@{}}\toprule
     & \multicolumn{5}{c}{\textbf{Greedy}} & \multicolumn{5}{c}{\textbf{CPLEX}} \\ \cmidrule{2-6} \cmidrule{7-11}
    \textbf{Virtual topology} & $\lfloor \frac{n_r}{10} \rfloor + 2$ & $\lfloor \frac{n_r}{10} \rfloor + 1$ & $\lfloor \frac{n_r}{10} \rfloor$ & $\lfloor \frac{n_r}{10} \rfloor -1 $ & $\lfloor \frac{n_r}{10} \rfloor -2$ & $\lfloor \frac{n_r}{10} \rfloor + 2$ & $\lfloor \frac{n_r}{10} \rfloor + 1$ & $\lfloor \frac{n_r}{10} \rfloor$ & $\lfloor \frac{n_r}{10} \rfloor-1$ & $\lfloor \frac{n_r}{10} \rfloor-2$ \\ \midrule 
  IRIS & 293 & 308 & 299 & 289 & 282 & 292 & 293 & 290 & 274 & 267 \\
  zib54 & 373 & 399 & 402 & 386 & 355 & 364 & 390 & 380 & 345 & 335 \\
  HiberniaGlobal& 358 & 372 & 366 & 379 & 353 & 350 & 357 & 333 & 333 & 320 \\
  US & 386 & 383 & 385 & 382 & 384 & 376 & 366 & 363 & 360 & 355 \\
  ta2 & 530 & 553 & 548 & 544 & 550 & 536 & 527 & 528 & 510 & 485 \\
  TW & 563 & 593 & 619 & 622 & 629 & 581 & 592 & 587 & 555 & 576 \\
  Intellifiber & 470 & 476 & 461 & 451 & 458 & 445 & 445 & 433 & 443 & 418 \\
  Uninett2010& 515 & 519 & 573 & 531 & 531 & 491 & 506 & 546 & 512 & 491 \\
  OTEGlobe & 510 & 501 & 520 & 502 & 479 & 489 & 478 & 488 & 471 & 444 \\
  Interroute & 844 & 818 & 793 & 809 & 780 & 801 & 800 & 771 & 762 & 745 \\
  \textbf{Average cost} & 484.2 & 492.2 & 496.6 & 489.5 & 480.1 & 472.5 & 475.4 & 471.9 & 456.5 & 443.6  \\
    \textbf{Time (s)} & 18.0 & 24.5 & 20.8 & 19.5 & 20.3 & 62.7 & 72.6 & 89.7 & 256.5 & 1970.8 \\
    \bottomrule
\end{tabular*}
\caption{Price-and-Branch heuristic performance depending 
on the pricer (first row) and the number of virtual subgraphs $k_r$ (second row)}
\label{table:nbsubg}
\end{table*}

\begin{table*}[t]
\footnotesize                      
\setlength{\tabcolsep}{3pt}        
\centering
\begin{tabular*}{\textwidth}{@{\extracolsep\fill}llllllllllllll@{}}\toprule
   & \multicolumn{6}{c}{\textbf{Greedy}} & \multicolumn{6}{c}{\textbf{CPLEX}} \\ \cmidrule{2-7}\cmidrule{8-13}
  \textbf{Virtual topology} & $k_r$ & 25 & 50 & 100 & 200 & 500 & $k_r$ & 25 & 50 & 100 & 200 & 500 \\ \midrule
  IRIS & 373 & 344 & 335 & 317 & 299 & 300 & 407 & 321 & 302 & 296 & 290 & 287 \\
  zib54 & 556 & 430 & 420 & 411 & 402 & 394 & 526 & 405 & 396 & 384 & 380 & 369 \\
  HiberniaGlobal& 520 & 411 & 386 & 373 & 366 & 349 & 472 & 380 & 362 & 344 & 333 & 324 \\
  US & 524 & 437 & 417 & 382 & 385 & 374 & 477 & 397 & 377 & 371 & 363 & 356 \\
  ta2 & 867 & 618 & 621 & 563 & 548 & 553 & 691 & 584 & 553 & 539 & 528 & 515 \\
  TW & 852 & 684 & 669 & 640 & 619 & 609 & 817 & 659 & 614 & 605 & 587 & 564 \\
  Intellifiber & 654 & 492 & 477 & 474 & 461 & 442 & 562 & 485 & 456 & 440 & 433 & 419 \\
  Uninett2010& 845 & 629 & 622 & 599 & 573 & 554 & 755 & 612 & 585 & 560 & 546 & 530 \\
  OTEGlobe & 737 & 565 & 544 & 527 & 520 & 505 & 602 & 532 & 512 & 496 & 488 & 476 \\
    Interroute & 911 & 899 & 908 & 812 & 793 & 760 & 928 & 857 & 809 & 793 & 771 & 734 \\
\textbf{Average cost} & 683.9 & 550.9 & 539.9 & 509.8 & 496.6 & 484.0 & 623.7 & 523.2 & 496.6 & 482.8 & 471.9 & 457.4 \\
  \textbf{Time (s)} & 0.9 & 1.9 & 3.3 & 5.5 & 20.8 & 237.1 & 2.9 & 10.9 & 20.0 & 41.6 & 89.7 & 360.1 \\
  \bottomrule
\end{tabular*}
\caption{Price-and-Branch heuristic performance depending 
on the pricer (first row) and number of columns (second row) (when there is $k_r$ columns, there is one column per virtual subgraph)}
\label{table:nbcols}
\end{table*}

\begin{algorithm}[b]
\caption{GenerateIntegerColumns}\label{alg1}
\begin{algorithmic}[1]
\Require VNE instance, dual variables, Virtual partition $\pi_r$, substrate partition $\pi_s$
\Ensure $k_r$ sub-mappings (one per virtual subgraph)

\State Initialize arrays $a[1..k_r],\; t_V[1..n_r]$

\For{$\GraphH_r^i \in \pi_r$}
  \State $a[i] \gets \GraphH_s^j \in \pi_s \setminus a$
  \State $u_c \gets $ Center$(a[i])$
  \State $t_V[\ubar] \gets u_c \;\; \forall \ubar \in V(\GraphH_r^i)$
\EndFor

\For{$\GraphH_r^i \in \pi_r$}
  \State $w'_{\ubar u} \gets w_u \;\; \forall \ubar \in V_r,\; u \in V_s$
  \For{$\ebar \in E_r^0$ with $t(\ebar) \in V_r^i$}
    \State $\hat{w}^{t(\ebar)}_{u} \gets w_{u} + $ ShortestPath$(t_V(s(\ebar)), u)$
  \EndFor
  \For{$\ebar \in E_r^0$ with $s(\ebar) \in V_r^i$}
    \State $\hat{w}^{s(\ebar)}_{u} \gets w_{u} +$ ShortestPath$(t_V(t(\ebar)), u)$
  \EndFor
  \State $m^i \gets$ SolveSubPricer($\GraphH_r^i, a[i], \hat{w}$, dual vars)
  \State $t_V[\ubar] \gets m^i_V(\ubar) \;\; \forall \ubar \in V_r^i$
\EndFor

\State \Return $\{m^i : \GraphH_r^i \in \pi_r\}$
\end{algorithmic}
\end{algorithm}

To tackle poor boundary placements, our key idea is to modify the objective functions of the sub-pricers to incorporate an estimation of routing costs for virtual cut edges.
This is done through Algorithm~\ref{alg1}, which, given the virtual and substrate partitions, generates $k_r$ columns (one per virtual subgraph) as follows.
Each virtual subgraph is assigned to a distinct substrate subgraph (array $a$), and virtual nodes have temporary placements (array $t_V$). At the beginning, this temporary placement corresponds to the most central substrate node of the substrate subgraph to which they are assigned to (function Center). To do this, we use closeness centrality defined in \cite{sabidussi1966centrality}. When embedding a subgraph, the cost of placing a boundary node is adjusted by the shortest-path distance (function ShortestPath) to the temporary placements of its neighbors in other substrate subgraphs. To this end, we use the Floyd-Warshall algorithm described in \cite{floyd1962algorithm}. This adjustment makes substrate boundary nodes more attractive for virtual boundary nodes, and generated columns (function SolveSubPricer) are more consistent with each other as they take into account the sub-mapping of other subgraphs. 
The assignment of a virtual subgraph to a substrate subgraph is randomly decided. \\

In our Price-and-Branch heuristic, Algorithm~\ref{alg1} is used until a given number of columns is generated.
Preliminary experiments showed that incorporating routing-cost estimations improves solution quality by 5–10\%, with no increase in solving time.

\subsection{Heuristics results}

We first observe the behavior of the Price-and-Branch algorithm with different tunings of the number of virtual subgraphs $k_r$ and the number of columns. Then we run tests over the whole instance set. 
\\

\begin{table*}
\centering
\footnotesize                      
\setlength{\tabcolsep}{3pt}        
\begin{tabular*}{\textwidth}{@{\extracolsep\fill}lllllllllll@{}}\toprule
    & & \multicolumn{5}{c}{\textbf{Fast (30 seconds)}} & \multicolumn{4}{c}{\textbf{Long (10 minutes)}} \\
    \cmidrule{3-7} \cmidrule{8-11} 
  \multicolumn{2}{c}{\textbf{Topology}} & \multicolumn{3}{c}{\textbf{Greedy}} & \multicolumn{2}{c}{\textbf{P\&B}} & \multicolumn{2}{c}{\textbf{CPLEX}} & \multicolumn{2}{c}{\textbf{P\&B}} \\ \cmidrule{1-2} \cmidrule{3-5} \cmidrule{6-7} \cmidrule{8-9} \cmidrule{10-11} 
  \textbf{Virtual} & \textbf{Substrate} & \textbf{Orig.} & \textbf{LS} & \textbf{PSO} & \textbf{Orig.} & \textbf{LS}  & \textbf{Orig.} & \textbf{LS}  & \textbf{Orig.} & \textbf{LS}  \\ 
  \midrule
     IRIS & TATA & 317.0 & 274.5 & \textbf{274.1} & 312.2 & 278.3 & 381.1 & \textbf{279.6} & 294.8 & 280.6 \\
   & GtsCe & 310.4 & 267.8 & 269.3 & 311.3 & \textbf{262.5} & 370.1 & 279.2 & 276.7 & \textbf{268.5} \\
   & Cogent & 297.1 & \textbf{261.7} & 265.9 & 295.4 & 266.3 & 382.7 & 264.3 & 269.5 & \textbf{258.8} \\
  zib54 & TATA & 443.7 & 370.3 & 375.6 & 407.4 & \textbf{361.4} & 640.5 & 372.5 & 394.0 & \textbf{366.2} \\
   & GtsCe & 429.3 & 359.6 & 364.5 & 397.8 & \textbf{346.8} & 585.5 & 360.4 & 370.2 & \textbf{357.6} \\
   & Cogent & 404.3 & \textbf{335.8} & 339.0 & 381.5 & 344.5 & 532.4 & 345.7 & 346.0 & \textbf{330.6} \\
  HiberniaGlobal& TATA & 382.6 & \textbf{330.1} & 341.5 & 370.6 & 333.1 & 596.5 & 347.9 & 340.5 & \textbf{331.2} \\
   & GtsCe & 376.5 & \textbf{327.3} & 336.3 & 361.2 & 331.7 & 550.0 & 346.8 & 338.5 & \textbf{328.9} \\
   & Cogent & 354.4 & 315.2 & 322.7 & 336.8 & \textbf{304.6} & 522.9 & 323.1 & 334.1 & \textbf{321.7} \\
  US & TATA & 392.9 & \textbf{339.3} & 347.7 & 377.0 & 340.3 & 538.3 & \textbf{343.1} & 384.3 & 348.1 \\
   & GtsCe & 390.6 & \textbf{339.7} & 343.7 & 375.6 & 343.3 & 587.0 & 353.7 & 369.8 & \textbf{345.5} \\
   & Cogent & 375.7 & \textbf{328.8} & 338.9 & 360.1 & 329.5 & 509.4 & 334.8 & 353.1 & \textbf{323.2} \\
  ta2 & TATA & 630.9 & 513.7 & 515.5 & 564.9 & \textbf{505.3} & 1007.4 & 521.6 & 565.2 & \textbf{498.7} \\
   & GtsCe & 611.9 & 498.1 & 503.1 & 552.0 & \textbf{495.7} & 1019.6 & 517.9 & 512.3 & \textbf{492.2} \\
   & Cogent & 577.2 & \textbf{472.1} & 474.7 & 537.3 & 473.5 & 1023.6 & 476.1 & 505.2 & \textbf{476.0} \\
  TW & TATA & 698.3 & 572.6 & 568.7 & 655.0 & \textbf{566.8} & 1216.8 & \textbf{581.8} & 656.1 & 584.9 \\
   & GtsCe & 682.4 & 566.2 & 573.7 & 645.9 & \textbf{564.5} & 1113.7 & 586.1 & 603.5 & \textbf{556.6} \\
   & Cogent & 639.3 & 520.1 & 517.4 & 581.9 & \textbf{516.8} & 1115.0 & 535.8 & 557.8 & \textbf{524.5} \\
  Intellifiber & TATA & 484.6 & 420.6 & 434.0 & 452.5 & \textbf{419.4} & 749.6 & 422.9 & 463.2 & \textbf{417.1} \\
   & GtsCe & 481.6 & 419.9 & 431.8 & 449.3 & \textbf{418.6} & 731.5 & 428.6 & 431.2 & \textbf{410.1} \\
   & Cogent & 458.7 & \textbf{396.5} & 400.4 & 438.8 & 401.2 & 769.6 & 417.3 & 428.0 & \textbf{397.1} \\
  Uninett2010& TATA & 604.6 & 499.6 & \textbf{496.6} & 623.1 & 502.8 & 956.3 & 510.9 & 559.4 & \textbf{486.5} \\
   & GtsCe & 580.9 & 489.4 & 480.6 & 603.6 & \textbf{479.5} & 1045.1 & 492.6 & 532.9 & \textbf{476.7} \\
   & Cogent & 543.4 & 450.9 & \textbf{449.4} & 553.9 & 455.5 & 836.6 & 464.8 & 495.1 & \textbf{453.9} \\
  OTEGlobe & TATA & 530.0 & \textbf{453.5} & 469.3 & 523.2 & 454.7 & 812.5 & 477.0 & 498.7 & \textbf{467.8} \\
   & GtsCe & 545.1 & 472.3 & 480.5 & 507.2 & \textbf{466.9} & 765.9 & 482.2 & 477.5 & \textbf{458.6} \\
   & Cogent & 508.5 & 436.9 & 446.1 & 501.0 & \textbf{435.2} & 801.6 & 446.2 & 457.3 & \textbf{436.5} \\
  Interroute & TATA & 872.7 & \textbf{715.1} & 738.6 & 810.7 & 729.6 & 2203.8 & 773.4 & 778.9 & \textbf{718.1} \\
   & GtsCe & 874.7 & 723.1 & 735.2 & 817.6 & \textbf{716.5} & 2079.2 & 753.7 & 749.5 & \textbf{703.4} \\
   & Cogent & 807.7 & 673.2 & 692.3 & 754.0 & \textbf{669.2} & 2159.0 & 703.4 & 721.4 & \textbf{664.3} \\*
   \multicolumn{2}{c}{\textbf{Average}} & 520.233 & 438.1 & 444.2 & 495.3 & \textbf{437.1} & 886.8 & 451.4 & 468.8 & \textbf{436.1} \\ \bottomrule
  \end{tabular*}
\caption{Price-Branch Heuristic results with large capacities}
\label{tab:large}
\end{table*}

\begin{table*}
\footnotesize                      
\setlength{\tabcolsep}{3pt}        
\centering
\begin{tabular*}{\textwidth}{@{\extracolsep\fill}lllllllllllllll@{}}\toprule
    \multicolumn{2}{c}{} & \multicolumn{7}{c}{\textbf{Fast (30 seconds)}} & \multicolumn{6}{c}{\textbf{Long (10 minutes)}} \\
    \cmidrule{3-9} \cmidrule{10-15} 
  \multicolumn{2}{c}{\textbf{Topology}} & \multicolumn{4}{c}{\textbf{Greedy}} & \multicolumn{3}{c}{\textbf{P\&B}}& \multicolumn{3}{c}{\textbf{CPLEX}}  & \multicolumn{3}{c}{\textbf{P\&B}}  \\ \cmidrule{1-2} \cmidrule{3-6} \cmidrule{7-9} \cmidrule{10-12} \cmidrule{13-15} 
  \textbf{Virtual} & \textbf{Substrate} & \textbf{\#feas} & \textbf{Orig.} & \textbf{LS} & \textbf{PSO} & \textbf{\#feas} & \textbf{Orig.} & \textbf{LS}  & \textbf{\#feas} &  \textbf{Orig.} & \textbf{LS}  & \textbf{\#feas} & \textbf{Orig.} & \textbf{LS}  \\ 
  \midrule
 IRIS & TATA & 10 & 318.1 & \textbf{278.9} & 286.6 & 10 & 315.1 & 284.7 & 10 & 404.3 & \textbf{278.1} & 10 & 295.8 & 283.7 \\
   & GtsCe & 10 & 313.5 & \textbf{269.0} & 274.8 & 10 & 313.2 & 269.3 & 10 & 371.6 & 274.9 & 10 & 276.2 & \textbf{266.7} \\
   & Cogent & 10 & 299.2 & \textbf{260.2} & 265.0 & 10 & 294.3 & 262.7 & 10 & 374.5 & 266.5 & 10 & 269.7 & \textbf{258.4} \\
  zib54 & TATA & 0 & &  & & 10 & 422.8 & \textbf{365.5} & 3 & 509.3 & \textbf{361.7} & 10 & 402.4 & 363.7 \\
   & GtsCe & 2 & 532.5 & \textbf{352.5} & 523.0 & 10 & 415.1 & 356.3 & 6 & 534.5 & 372.0 & 10 & 382.0 & \textbf{360.5} \\
   & Cogent & 2 & 513.5 & 356.5 & 463.0 & 10 & 402.5 & \textbf{338.3} & 6 & 513.8 & 342.2 & 10 & 356.6 & \textbf{340.6} \\
  HiberniaGlobal& TATA & 10 & 395.1 & \textbf{338.0} & 358.6 & 10 & 375.0 & 343.6 & 10 & 561.0 & 345.1 & 10 & 343.6 & \textbf{329.5} \\
   & GtsCe & 10 & 387.4 & \textbf{331.2} & 335.6 & 10 & 364.7 & 334.4 & 9 & 564.9 & 350.3 & 10 & 339.1 & \textbf{333.8} \\
   & Cogent & 10 & 370.1 & 317.8 & 324.2 & 10 & 340.9 & \textbf{309.7} & 6 & 568.8 & 322.3 & 10 & 340.2 & \textbf{319.6} \\
  US & TATA & 10 & 394.1 & \textbf{338.9} & 344.8 & 10 & 378.9 & 341.7 & 10 & 539.4 & 349.7 & 10 & 388.8 & \textbf{344.4} \\
   & GtsCe & 10 & 387.5 & \textbf{334.8} & 344.9 & 10 & 375.1 & 342.6 & 10 & 538.0 & 357.0 & 10 & 371.6 & \textbf{343.4} \\
   & Cogent & 10 & 377.7 & \textbf{330.2} & 337.8 & 10 & 361.9 & 335.5 & 10 & 527.0 & 339.2 & 10 & 355.5 & \textbf{324.0} \\
  ta2 & TATA & 0 & &  & & 5 & 660.0 & \textbf{525.4} & 0 & &  & 10 & 607.1 & \textbf{523.4} \\
   & GtsCe & 0 & &  & & 2 & 615.5 & \textbf{503.0} & 1 & 747.0 & 537.0 & 10 & 557.8 & \textbf{506.0} \\
   & Cogent & 0 & & &  & 0 & & & 1 & 733.0 & \textbf{478.0} & 10 & 544.4 & 484.8 \\
  TW & TATA & 0 & & &  & 0 & & & 0 & &  & 5 & 831.0 & \textbf{608.4} \\
   & GtsCe & 0 & & &  & 0 & & & 0 & & & 0 & &  \\
   & Cogent & 0 & & &  & 0 & & & 0 & &  & 8 & 696.7 & \textbf{563.6} \\
  Intellifiber & TATA & 10 & 502.9 & 422.8 & 445.0 & 10 & 456.7 & \textbf{416.7} & 9 & 630.8 & 436.8 & 10 & 471.3 & \textbf{416.3} \\
   & GtsCe & 10 & 493.4 & 416.4 & 430.2 & 10 & 449.5 & \textbf{414.2} & 9 & 654.3 & 430.0 & 10 & 437.4 & \textbf{409.8} \\
   & Cogent & 10 & 472.8 & \textbf{399.1} & 430.2 & 10 & 443.4 & 401.5 & 10 & 657.5 & 410.8 & 10 & 429.7 & \textbf{403.3} \\
  Uninett2010& TATA & 0 & & &  & 0 & & & 4 & 767.7 & 501.0 & 10 & 587.3 & \textbf{490.6} \\
   & GtsCe & 0 & & &  & 0 & & & 7 & 739.9 & 495.3 & 10 & 563.0 & \textbf{482.8} \\
   & Cogent & 0 & & &  & 0 & & & 5 & 761.8 & 496.4 & 10 & 532.8 & \textbf{464.9} \\
  OTEGlobe & TATA & 10 & 544.2 & \textbf{452.2} & 471.2 & 10 & 525.5 & 454.5 & 10 & 700.5 & 476.6 & 10 & 498.3 & \textbf{465.7} \\
   & GtsCe & 10 & 556.7 & 474.6 & 486.2 & 10 & 514.5 & \textbf{462.7} & 10 & 692.6 & 484.2 & 10 & 481.9 & \textbf{456.2} \\
   & Cogent & 10 & 526.7 & 443.4 & 473.0 & 10 & 513.8 & \textbf{432.1} & 10 & 660.9 & 450.5 & 10 & 458.1 & \textbf{438.2} \\
  Interroute & TATA & 0 & &  & & 10 & 892.7 & \textbf{745.1} & 2 & 1070.5 & 750.5 & 10 & 777.5 & \textbf{731.0} \\
   & GtsCe & 0 & &  & & 7 & 907.0 & \textbf{733.7} & 0 & &  & 10 & 758.2 & \textbf{703.6} \\
   & Cogent & 0 & &  & & 7 & 835.7 & \textbf{686.7} & 3 & 1073.7 & 683.3 & 10 & 725.3 & \textbf{650.6} \\\bottomrule
  \end{tabular*}
  \caption{Price-Branch Heuristic results with medium capacities}
\label{tab:medium}
\end{table*}

\begin{table*}
\centering
\footnotesize                      
\setlength{\tabcolsep}{3pt}        
\begin{tabular*}{\textwidth}{@{\extracolsep\fill}lllllllllllllll@{}}\toprule
    \multicolumn{2}{c}{} & \multicolumn{7}{c}{\textbf{Fast (30 seconds)}} & \multicolumn{6}{c}{\textbf{Long (10 minutes)}} \\
    \cmidrule{3-9} \cmidrule{10-15} 
  \multicolumn{2}{c}{\textbf{Topology}} & \multicolumn{4}{c}{\textbf{Greedy}} & \multicolumn{3}{c}{\textbf{P\&B}}& \multicolumn{3}{c}{\textbf{CPLEX}}  & \multicolumn{3}{c}{\textbf{P\&B}}  \\ \cmidrule{1-2} \cmidrule{3-6} \cmidrule{7-9} \cmidrule{10-12} \cmidrule{13-15} 
  \textbf{Virtual} & \textbf{Substrate} & \textbf{\#feas} & \textbf{Orig.} & \textbf{LS} & \textbf{PSO} & \textbf{\#feas} & \textbf{Orig.} & \textbf{LS}  & \textbf{\#feas} &  \textbf{Orig.} & \textbf{LS}  & \textbf{\#feas} & \textbf{Orig.} & \textbf{LS}  \\ 
  \midrule
  IRIS & TATA & 6 & 424.3 & \textbf{308.8} & 418.0 & 10 & 394.3 & 317.7 & 8 & 424.7 & 314.1 & 10 & 346.2 & \textbf{311.2} \\
   & GtsCe & 10 & 369.9 & \textbf{300.3} & 366.0 & 10 & 381.5 & 321.4 & 10 & 414.9 & 313.5 & 10 & 333.2 & \textbf{312.2} \\
   & Cogent & 10 & 369.8 & \textbf{297.3} & 358.3 & 10 & 338.3 & 300.3 & 8 & 436.2 & 311.7 & 10 & 313.4 & \textbf{291.5} \\
  zib54 & TATA & 0 &  &  &  & 9 & 628.4 & \textbf{417.0} & 0 &  &  & 10 & 514.8 & \textbf{421.2} \\
   & GtsCe & 0 &  &  &  & 0 &  &  & 0 &  &  & 8 & 525.0 & \textbf{424.2} \\
   & Cogent & 0 &  &  &  & 1 & 600.0 & \textbf{403.0} & 0 &  &  & 10 & 482.3 & \textbf{414.9} \\
  HiberniaGlobal& TATA & 0 &  &  &  & 10 & 479.7 & \textbf{391.2} & 2 & 500.0 & 388.5 & 10 & 434.5 & \textbf{387.3} \\
   & GtsCe & 0 &  &  &  & 10 & 432.0 & \textbf{381.6} & 3 & 515.3 & 391.3 & 9 & 405.2 & \textbf{374.3} \\
   & Cogent & 0 &  &  &  & 10 & 438.6 & \textbf{388.3} & 2 & 515.5 & 389.5 & 10 & 415.1 & \textbf{386.9} \\
  US & TATA & 6 & 513.8 & \textbf{392.3} & 492.2 & 10 & 465.5 & 392.4 & 3 & 572.0 & 426.0 & 10 & 465.0 & \textbf{403.3} \\
   & GtsCe & 10 & 488.0 & 391.2 & 473.3 & 10 & 446.9 & \textbf{391.0} & 5 & 594.0 & 407.4 & 10 & 453.9 & \textbf{406.8} \\
   & Cogent & 8 & 497.1 & 381.4 & 489.0 & 10 & 412.8 & \textbf{377.5} & 8 & 530.4 & 395.2 & 10 & 403.3 & \textbf{379.0} \\
  ta2 & TATA & 0 &  &  &  & 0 &  &  & 0 &  &  & 0 &  &  \\
   & GtsCe & 0 &  &  &  & 0 &  &  & 0 &  &  & 0 &  &  \\
   & Cogent & 0 &  &  &  & 0 &  &  & 0 &  &  & 0 &  &  \\
  TW & TATA & 0 &  &  &  & 0 &  &  & 0 &  &  & 0 &  &  \\
   & GtsCe & 0 &  &  &  & 0 &  &  & 0 &  &  & 0 &  &  \\
   & Cogent & 0 &  &  &  & 0 &  &  & 0 &  &  & 0 &  &  \\
  Intellifiber & TATA & 0 &  &  &  & 9 & 636.2 & \textbf{513.3} & 1 & 666.0 & 497.0 & 10 & 594.6 & \textbf{489.1} \\
   & GtsCe & 0 &  &  &  & 10 & 551.6 & \textbf{496.0} & 1 & 656.0 & 542.0 & 9 & 563.0 & \textbf{489.0} \\
   & Cogent & 0 &  &  &  & 10 & 510.7 & \textbf{462.1} & 1 & 619.0 & 487.0 & 10 & 504.5 & \textbf{468.0} \\
  Uninett2010& TATA & 0 &  &  &  & 0 &  &  & 0 &  &  & 0 &  &  \\
   & GtsCe & 0 &  &  &  & 0 &  &  & 0 &  &  & 0 &  &  \\
   & Cogent & 0 &  &  &  & 0 &  &  & 0 &  &  & 0 &  &  \\
  OTEGlobe & TATA & 0 &  &  &  & 7 & 692.3 & \textbf{564.3} & 0 &  &  & 10 & 610.9 & \textbf{549.1} \\
   & GtsCe & 0 &  &  &  & 7 & 680.0 & \textbf{544.9} & 4 & 705.7 & 555.7 & 10 & 563.7 & \textbf{532.7} \\
   & Cogent & 0 &  &  &  & 10 & 602.4 & \textbf{505.2} & 1 & 673.0 & 506.0 & 10 & 541.5 & \textbf{492.0} \\
  Interroute & TATA & 0 &  &  &  & 0 &  &  & 0 &  &  & 0 &  &  \\
   & Cogent & 0 &  &  &  & 0 &  &  & 0 &  &  & 2 & 876.5 & \textbf{802.5} \\ \bottomrule
  \end{tabular*}
    \caption{Price-Branch Heuristic results with small capacities}
\label{tab:small}
\end{table*}

\begin{table*}
\centering
\footnotesize                      
\setlength{\tabcolsep}{3pt}        

\begin{tabular*}{\textwidth}{@{\extracolsep\fill}lllllllllllllll@{}}\toprule
    \multicolumn{2}{c}{} & \multicolumn{7}{c}{\textbf{Fast (30 seconds)}} & \multicolumn{6}{c}{\textbf{Long (10 minutes)}} \\
    \cmidrule{3-9} \cmidrule{10-15} 
  \multicolumn{2}{c}{\textbf{Topology}} & \multicolumn{4}{c}{\textbf{Greedy}} & \multicolumn{3}{c}{\textbf{P\&B}}& \multicolumn{3}{c}{\textbf{CPLEX}}  & \multicolumn{3}{c}{\textbf{P\&B}}  \\ \cmidrule{1-2} \cmidrule{3-6} \cmidrule{7-9} \cmidrule{10-12} \cmidrule{13-15} 
  \textbf{Virtual} & \textbf{Substrate} & \textbf{\#feas} & \textbf{Orig.} & \textbf{LS} & \textbf{PSO} & \textbf{\#feas} & \textbf{Orig.} & \textbf{LS}  & \textbf{\#feas} &  \textbf{Orig.} & \textbf{LS}  & \textbf{\#feas} & \textbf{Orig.} & \textbf{LS}  \\ 
  \midrule
 IRIS & TATA & 10 & 846.5 & 615.3 & 812.7 & 10 & 737.7 & \textbf{609.3} & 10 & 617.2 & \textbf{586.3} & 10 & 645.2 & 609.8 \\
   & GtsCe & 8 & 839.2 & \textbf{612.0} & 785.1 & 10 & 762.5 & 642.5 & 10 & 632.1 & \textbf{601.9} & 10 & 668.9 & 622.9 \\
   & Cogent & 7 & 852.9 & \textbf{569.0} & 795.3 & 10 & 714.5 & 584.0 & 10 & 754.8 & 591.5 & 10 & 627.0 & \textbf{572.8} \\
  zib54 & TATA & 0 &  &  &  & 10 & 1055.8 & \textbf{861.2} & 4 & 1092.7 & 881.0 & 10 & 897.8 & \textbf{855.8} \\
   & GtsCe & 0 &  &  &  & 10 & 1026.3 & \textbf{877.2} & 3 & 1134.0 & \textbf{825.3} & 10 & 896.3 & 843.4 \\
   & Cogent & 0 &  &  &  & 10 & 927.2 & \textbf{780.2} & 3 & 1076.3 & 782.3 & 10 & 776.2 & \textbf{741.0} \\
  HiberniaGlobal& TATA & 1 & 1171.0 & 947.0 & 1171.0 & 9 & 916.3 & \textbf{803.2} & 4 & 1159.2 & 803.2 & 9 & 813.2 & \textbf{790.9} \\
   & GtsCe & 4 & 1115.7 & \textbf{720.0} & 1027.7 & 10 & 875.5 & 779.9 & 4 & 1062.7 & 771.2 & 10 & 773.4 & \textbf{742.8} \\
   & Cogent & 1 & 1004.0 & \textbf{693.0} & 1004.0 & 10 & 880.6 & 756.5 & 6 & 1080.2 & 764.8 & 9 & 759.6 & \textbf{726.6} \\
  US & TATA & 8 & 1032.6 & \textbf{737.9} & 1002.5 & 10 & 867.7 & 752.9 & 9 & 999.4 & 756.7 & 10 & 854.4 & \textbf{753.5} \\
   & GtsCe & 9 & 1008.6 & \textbf{722.2} & 938.7 & 10 & 853.8 & 731.0 & 9 & 986.3 & \textbf{734.7} & 10 & 804.4 & 748.9 \\
   & Cogent & 2 & 1044.0 & \textbf{729.5} & 1044.0 & 10 & 843.8 & 741.4 & 10 & 1083.1 & 745.3 & 10 & 794.0 & \textbf{731.5} \\
  ta2 & TATA & 0 &  &  &  & 7 & 1436.7 & \textbf{1222.9} & 0 &  &  & 10 & 1407.2 & \textbf{1306.2} \\
   & GtsCe & 0 &  &  &  & 3 & 1479.7 & \textbf{1137.3} & 0 &  &  & 9 & 1256.3 & \textbf{1198.4} \\
   & Cogent & 0 &  &  &  & 4 & 1371.0 & \textbf{1213.5} & 0 &  &  & 10 & 1205.4 & \textbf{1166.4} \\
  TW & TATA & 0 &  &  &  & 0 &  &  & 0 &  &  & 5 & 1575.2 & \textbf{1394.0} \\
   & GtsCe & 0 &  &  &  & 0 &  &  & 0 &  &  & 6 & 1449.8 & \textbf{1318.2} \\
   & Cogent & 0 &  &  &  & 0 &  &  & 0 &  &  & 7 & 1467.4 & \textbf{1364.9} \\
  Intellifiber & TATA & 1 & 1551.0 & 1014.0 & 1551.0 & 10 & 1084.6 & \textbf{965.5} & 5 & 1387.6 & 1017.4 & 10 & 1090.8 & \textbf{973.0} \\
   & GtsCe & 2 & 1278.0 & \textbf{897.0} & 1278.0 & 9 & 1058.3 & 942.7 & 3 & 1333.7 & 989.7 & 9 & 1002.0 & \textbf{936.3} \\
   & Cogent & 0 &  &  &  & 10 & 1039.7 & \textbf{924.1} & 7 & 1349.1 & 969.1 & 10 & 963.2 & \textbf{897.1} \\
  Uninett2010& TATA & 0 &  &  &  & 2 & 1808.5 & \textbf{1148.5} & 1 & 1565.0 & 1162.0 & 9 & 1255.7 & \textbf{1141.3} \\
   & GtsCe & 0 &  &  &  & 0 &  &  & 0 &  &  & 10 & 1281.7 & \textbf{1122.4} \\
   & Cogent & 0 &  &  &  & 0 &  &  & 0 &  &  & 10 & 1191.7 & \textbf{1060.3} \\
  OTEGlobe & TATA & 1 & 1199.0 & \textbf{884.0} & 1199.0 & 10 & 1245.3 & 1027.6 & 9 & 1380.6 & 1027.9 & 10 & 1088.0 & \textbf{1013.9} \\
   & GtsCe & 1 & 1438.0 & \textbf{1054.0} & 1438.0 & 10 & 1274.1 & 1061.2 & 4 & 1392.2 & 1064.0 & 10 & 1114.3 & \textbf{1051.9} \\
   & Cogent & 0 &  &  &  & 10 & 1155.8 & \textbf{962.4} & 8 & 1297.1 & 985.7 & 10 & 1001.7 & \textbf{948.6} \\
  Interroute & TATA & 0 &  &  &  & 2 & 1928.0 & \textbf{1570.5} & 0 &  &  & 5 & 1711.8 & \textbf{1576.0} \\
   & GtsCe & 0 &  &  &  & 4 & 2003.2 & \textbf{1624.5} & 0 &  &  & 8 & 1742.0 & \textbf{1601.4} \\
   & Cogent & 0 &  &  &  & 3 & 1970.3 & \textbf{1582.7} & 1 & 2102.0 & 1660.0 & 10 & 1696.9 & \textbf{1508.8} \\ \bottomrule
  \end{tabular*}  
  \caption{Price-Branch Heuristic results with random demands}
\label{tab:random}
\end{table*}

\subsubsection{Preliminary analysis}

In what follows, the tests are run using the Cogent substrate network, with large capacities.
First, we study the influence of the partition size $k_r$ on the Price-and-Branch heuristic, with 200 columns generated by greedy or exact pricers (CPLEX with the Flow Formulation), with Algorithm~\ref{alg1}. The results are reported in Table~\ref{table:nbsubg}.
With greedy pricers, solving times and solution quality remain quite stable. Greedy sub-pricers are quite competitive for high $k_r$, since near-optimal solutions can be found with small virtual subgraphs,. 
On the other hand, for exact pricers, we find better solutions with $k_r \ge \lfloor \frac{n_r}{10} \rfloor$. The best solutions are found with $k_r = \lfloor \frac{n_r}{10} \rfloor - 2$, with, unfortunately, very long solving times, as the sub-pricers become more difficult. \\ 

Next, we study the impact of the number of columns. The results are reported in Table~\ref{table:nbcols}. 
As expected, using more columns leads to better solutions. Since the solving time of the integer Restricted Master Problem increases exponentially, we used a time limit of 600 seconds, which was often met with 500 columns.
Again, using exact pricers, compared to greedy algorithms, improves the performance on average by 5 to 10\%, at the cost of higher solving times. 

\subsubsection{Algorithmic Setup}

We now evaluate the performance of the proposed Price\&Branch (P\&B) framework on the full set of instances against alternative state-of-the-art methods. Algorithms for VNE fall into two categories: constructive algorithms that generate an initial embedding from scratch (e.g., greedy heuristics), and improvement algorithms that iteratively refine an existing solution (e.g., metaheuristics). Our Price\&Branch heuristic is a constructive algorithm. To evaluate it against other constructive algorithms, we study two distinct solving time horizons. First, a short one (30-second time limit) with:
\begin{itemize}
    \item The Price\&Branch heuristic with $k_r = \lfloor \frac{n_r}{10} \rfloor$ virtual subgraphs, greedy sub-pricers and 250 columns.
    \item The greedy heuristic Algorithm~\ref{alg:greedy}, executed repeatedly for 30 seconds, leading to tens of thousand of iterations. This greedy heuristic is similar to those of \cite{cheng2011virtual, gong2014toward, zhang2016virtual}, but optimized for our offline case, as described earlier. 
\end{itemize}

Second, a long one (10-minute time limit) with:
\begin{itemize}
    \item The Price\&Branch heuristic with $k_r = \lfloor \frac{n_r}{11.5} \rfloor$ virtual subgraphs, exact (CPLEX) sub-pricers and 500 columns. If there is not enough time to generate all the columns with CPLEX, greedy sub-pricers are used to fill the remaining number of columns, though this rarely occurs. Note that, in many cases, the solving time is quite smaller than 10 minutes, if the columns are generated quickly or the RMP optimization is faster than the one-minute allocated to it.
    \item The CPLEX solver with the Flow Formulation in heuristic mode.
\end{itemize}

Once the initial solution is generated, it is improved iteratively for 30 seconds with:
\begin{itemize}
    \item the Local Search (LS) Algorithm~\ref{alg:localsearch}, which removes $n_{LS}$ adjacent nodes of a solution and relocates them with the same greedy procedure as in Algorithm~\ref{alg:greedy}. Virtual edges for which one or two ends have been relocated are then rerouted using Successive Shortest Path Routings. This algorithm can be seen as a Large Neighborhood Search (LNS) \cite{pisinger2018large}. Note that Local Search algorithms have received little attention in the VNE literature \cite{infuhr2013solving,song2019divide,rguez2025grasp}. In \cite{rguez2025grasp}, a similar procedure is used in a GRASP method, but only one node is relocated. However, in early tests, we found that removing a larger number of nodes leads to much better results for our large instances: we use $n_{LS} = 15$. In \cite{song2019divide}, the virtual network is also partitioned into subgraphs of around 10 nodes, and an entire subgraph is remapped. However, removing a random part of the virtual network, as we do, can lead to more variety in the neighborhood search.
    \item The Particle Swarm Optimization (PSO) algorithm from \cite{zhang2013unified}. In the PSO, the node relocation is achieved by the same greedy procedure as in Algorithm~\ref{alg:greedy}, and the edge routing by a Successive Shortest Path Routing. This metaheuristic is used with 25 particles, which are the 25 best solutions found previously by the greedy algorithm. Note that, due to the random nature of the greedy heuristic, these solutions are very diversified, which is desirable. The PSO algorithm is not used with other initial solution algorithms, since they can only offer one solution. 
Note that, for both algorithms, 30 seconds is more than enough to converge.
\end{itemize}

\begin{algorithm}
\fontsize{9.5pt}{11pt}\selectfont 
\caption{Local Search}\label{alg:localsearch}
\begin{algorithmic}[1]
\Require VNE instance, a feasible mapping $m$
\Ensure a new mapping (if one was found)
\State Initialize array $m'_V = m_V$, $m'_E = m_E$
\State $V'_r = \text{random}(V_r)$
\While{$|V'_r| \le n_{LS}$}
  \State $V'_r \gets V'_r \cup \text{random}(N(V'_r))$
\EndWhile
\For{$\bar{u} \in V'_r$} 
    \State $m'_V(\bar{u}) = 0$
\EndFor
\For{$\bar{e} \in E_r$} 
    \If{$s(\bar{e}) \in V'_r $ or $ t(\bar{e}) \in V'_r$}
        \State $m'_E(\bar{e}) = ()$
    \EndIf
\EndFor
\State $m'_V = \text{CompleteNodePlacement}(m'_V)$
\State $m'_E = \text{CompleteShortestPathRouting}(m'_V, m'_E)$
\State \Return $m'=(m'_V, m'_E)$
\end{algorithmic}
\end{algorithm}

\subsubsection{Overall results}

The results are for large (resp. medium, small) capacities are summarized in Table~\ref{tab:large} (resp. Table~\ref{tab:medium}, Table~\ref{tab:small}). For each algorithm, 10 runs are achieved. If an algorithm fails to produce a solution, then the corresponding cell is empty. For large capacities, all algorithms are able to find solutions, since any node placement induces a feasible mapping. For medium and small capacities, we report the number of times a solution is found by the algorithm. 
The best result found for a time-limit is highlighted in bold.
Overall, the experiments show that our algorithm is the only one able to solve difficult instances. In what follows, we discuss in more detail the results.
\\

\noindent \textbf{Finding a feasible solution:} We first discuss the case of large capacities. Overall, our Price\&Branch algorithm finds better initial solutions than the greedy algorithm and CPLEX. In 30 seconds, the greedy heuristic sometimes outperforms the corresponding P\&B. 
However, this seems to depend on the virtual network: on ta2, the P\&B consistently proposes better solutions than the greedy algorithm, but on TW it is the opposite.
The fact that it depends so much on the virtual network suggests that this relatively poor performance of the P\&B is due to a poor initial partitioning.
The P\&B performs better with 10 minutes, although in a few cases, the 30-second variant is the best.  Finally, CPLEX proposes very poor solutions for all instances.
We now discuss the case of tighter capacities. Here, the P\&B is often the only algorithm to propose a solution, especially in 10 minutes. On the other hand, the greedy algorithm often fails to find a solution, whereas CPLEX performs better. Overall, these results show that our subgraph-to-subgraph strategy is efficient, especially for difficult instances. \\ 


\noindent \textbf{Solution refinement:}
The results show that the Local Search algorithm is very efficient in all cases. Even with a much worse baseline, it is able to recover good solutions (e.g., with CPLEX initial solutions in the large capacity scenario). Starting from a better initial solution helps slightly, as the P\&B with 10 minutes still ends up with the best solutions after the Local Search, but only very slightly better than the much faster greedy algorithm.
Overall, LS outperforms PSO, especially when capacities become sparse, in which case PSO often fails to provide any improvement to the initial solution. We believe this is due to the large size of the instances, for which PSO is not able to explore the solution space efficiently.
Note that LS can reduce the gap to the best lower bound known (from Section~\ref{sec:lowerbounds}) to less than 30\% in several cases.
Overall, these results show that the Local Search algorithm proposed is very adapted for large instances. Moreover, for easy instances, using the greedy algorithm as an initial solution is enough. However, our P\&B algorithm becomes useful for difficult instances. \\ 

\noindent \textbf{Instance difficulty:} 
It appears that dense virtual networks (e.g., ta2 or TW) are much more difficult to embed when capacities become sparse. This is not a surprise, as the edge routing becomes a major bottleneck to finding a feasible mapping. For these instances, the Price\&Branch heuristic is often the only one to find a feasible solution. In several instances, the feasibility remains unknown. In a few cases (TW on GtsCe with medium capacities, Uninett2010 on TATA with low capacities), we have, throughout our tests, already found solutions, and Uninett2010 on GtsCe with small capacities was proven unfeasible. This leaves 9 instances for which the feasibility is an open question.


\subsection{Results with random demands}

Finally, we test the algorithms with random demands, to verify that the behavior remains the same on this more general configuration. We retain the same algorithms, which do not take demands into account, for fair comparison. Note that each algorithm could be improved for non-uniform demands: in particular, our Price\&Branch could aim at different partitions depending on the demands of nodes and edges. For each instance (a virtual and a substrate network), 10 random instances are generated as follows. On each virtual node and edge, the (integer) demand is generated randomly between 1 and 3. On each substrate node (resp. edge), the (integer) capacity is between 1 and 3 (resp. 4 and 12). The costs are those of the original instances.
\\

The results are summarized in Table~\ref{tab:random}. With the random configuration chosen, it appears that the instances are difficult, as the greedy algorithm often fails to achieve an initial embedding. Overall, the algorithms show similar behavior to that observed under the uniform demand case with medium and small capacities. For the dense virtual networks, such as TW or Uninett2010, the Price\&Branch heuristic is often the only one able to find a solution. The Local Search algorithm is again able to largely improve the solutions, contrary to PSO.

\section{Conclusion}

This paper introduces a new integer linear programming formulation for Virtual Network Embedding, based on a partition of the virtual network.
We propose an automatic decomposition scheme that allows a good trade-off between the size of the partition and the efficiency of the pricing problems.
A first column generation approach using dedicated heuristic pricers is able to compute better lower bounds than state-of-the-art methods.
In the aim of producing solutions, we identified some column structures and devised a Price-and-Branch heuristic that obtains solutions for instances with large virtual networks, up to 110 nodes in our tests. 
We then proposed a Local Search algorithm, which greatly improves the solutions obtained.
Experiments show that our subgraph-to-subgraph heuristic
outperforms the algorithms from the literature for large instances of the Internet Topology Zoo and SNDlib, in particular with tight capacities instances, where traditional approaches are not able to obtain any solutions. 
Moreover, the Local Search is highly efficient, allowing us to achieve integrality gaps under 30\% with the bounds we produced in several cases.
\\

Within the instance benchmark we have compiled from the literature, some instances are still unsolved,
suggesting avenues to further refine to our subgraph-to-subgraph heuristics.
On the other hand, stronger lower bounds could be produced by
leveraging the decomposition using other partitions where the virtual subgraphs can overlap.
Since the pricers are based on the Flow Formulation, a reinforcement through valid inequalities based on a polyhedral study will allow improvements in the efficiency of our method.
These improvements could eventually lead to a Branch\&Price capable of solving large VNE instances to optimality.

\bibliographystyle{unsrt}  
\bibliography{biblio}  

@article{anderson2005overcoming,
  title={Overcoming the Internet impasse through virtualization},
  author={Anderson, Thomas and Peterson, Larry and Shenker, Scott and Turner, Jonathan},
  journal={Computer},
  volume={38},
  number={4},
  pages={34--41},
  year={2005},
  publisher={IEEE}
}

@inproceedings{turner2005diversifying,
  title={Diversifying the internet},
  author={Turner, Jonathan S and Taylor, David E},
  booktitle={GLOBECOM'05. IEEE Global Telecommunications Conference, 2005.},
  volume={2},
  pages={6--pp},
  year={2005},
  organization={IEEE}
}

@article{chowdhury2010survey,
  title={A survey of network virtualization},
  author={Chowdhury, NM Mosharaf Kabir and Boutaba, Raouf},
  journal={Computer Networks},
  volume={54},
  number={5},
  pages={862--876},
  year={2010},
  publisher={Elsevier}
}

@article{alliance2016description,
  title={Description of network slicing concept},
  author={Alliance, NGMN},
  journal={NGMN 5G P},
  volume={1},
  number={1},
  pages={1--11},
  year={2016}
}

@article{3gpp2018slicing,
  title={5G; Management and orchestration; Concepts, use cases and requirements, Release 15, TS 28.530},
  author={3rd Generation Partnership Project},
  year={2018}
}

@article{vassilaras2017algorithmic,
  title={The algorithmic aspects of network slicing},
  author={Vassilaras, Spyridon and Gkatzikis, Lazaros and Liakopoulos, Nikolaos and Stiakogiannakis, Ioannis N and Qi, Meiyu and Shi, Lei and Liu, Liu and Debbah, Merouane and Paschos, Georgios S},
  journal={IEEE Communications Magazine},
  volume={55},
  number={8},
  pages={112--119},
  year={2017},
  publisher={IEEE}
}

@article{afolabi2018network,
  title={Network slicing and softwarization: A survey on principles, enabling technologies, and solutions},
  author={Afolabi, Ibrahim and Taleb, Tarik and Samdanis, Konstantinos and Ksentini, Adlen and Flinck, Hannu},
  journal={IEEE Communications Surveys \& Tutorials},
  volume={20},
  number={3},
  pages={2429--2453},
  year={2018},
  publisher={IEEE}
}

@article{foukas2017network,
  title={Network slicing in 5G: Survey and challenges},
  author={Foukas, Xenofon and Patounas, Georgios and Elmokashfi, Ahmed and Marina, Mahesh K},
  journal={IEEE communications magazine},
  volume={55},
  number={5},
  pages={94--100},
  year={2017},
  publisher={IEEE}
}

@article{ordonez2017network,
  title={Network slicing for 5G with SDN/NFV: Concepts, architectures, and challenges},
  author={Ordonez-Lucena, Jose and Ameigeiras, Pablo and Lopez, Diego and Ramos-Munoz, Juan J and Lorca, Javier and Folgueira, Jesus},
  journal={IEEE Communications Magazine},
  volume={55},
  number={5},
  pages={80--87},
  year={2017},
  publisher={IEEE}
}

@online{Tmobile2025priority,
  author = {T-Mobile},
  title = {America’s best 5G network experience for first responders},
  year = 2025,
  url = {https://www.t-mobile.com/business/t-priority},
  urldate = {2025-08-20}
}

@online{Siemens2025water,
  author = {Siemens },
  title = {Siemens and O2 Telefónica partner to develop solutions based on 5G network slicing},
  year = 2025,
  url = {https://press.siemens.com/global/en/pressrelease/siemens-and-o2-telefonica-partner-develop-solutions-based-5g-network-slicing},
  urldate = {2025-08-20}
}

@article{fischer2013virtual,
  title={Virtual network embedding: A survey},
  author={Fischer, Andreas and Botero, Juan Felipe and Beck, Michael Till and De Meer, Hermann and Hesselbach, Xavier},
  journal={IEEE Communications Surveys \& Tutorials},
  volume={15},
  number={4},
  pages={1888--1906},
  year={2013},
  publisher={IEEE}
}

@article{amaldi2016computational,
  title={On the computational complexity of the virtual network embedding problem},
  author={Amaldi, Edoardo and Coniglio, Stefano and Koster, Arie MCA and Tieves, Martin},
  journal={Electronic Notes in Discrete Mathematics},
  volume={52},
  pages={213--220},
  year={2016},
  publisher={Elsevier}
}

@article{benhamiche2025complexity,
  title={Complexity of the Virtual Network Embedding with uniform demands},
  author={Benhamiche, Amal and Fouilhoux, Pierre and L{\'e}tocart, Lucas and Perrot, Nancy and Schneider, Alexis},
  journal={arXiv preprint arXiv:2501.10154},
  year={2025}
}

@article{rost2020hardness,
  title={On the hardness and inapproximability of virtual network embeddings},
  author={Rost, Matthias and Schmid, Stefan},
  journal={IEEE/ACM transactions on networking},
  volume={28},
  number={2},
  pages={791--803},
  year={2020},
  publisher={IEEE}
}

@article{rost2015beyond,
  title={Beyond the stars: Revisiting virtual cluster embeddings},
  author={Rost, Matthias and Fuerst, Carlo and Schmid, Stefan},
  journal={ACM SIGCOMM Computer Communication Review},
  volume={45},
  number={3},
  pages={12--18},
  year={2015},
  publisher={ACM New York, NY, USA}
}

@inproceedings{chowdhury2009virtual,
  title={Virtual network embedding with coordinated node and link mapping},
  author={Chowdhury, NM Mosharaf Kabir and Rahman, Muntasir Raihan and Boutaba, Raouf},
  booktitle={IEEE INFOCOM 2009},
  pages={783--791},
  year={2009},
  organization={IEEE}
}

@article{chowdhury2011vineyard,
  title={Vineyard: Virtual network embedding algorithms with coordinated node and link mapping},
  author={Chowdhury, Mosharaf and Rahman, Muntasir Raihan and Boutaba, Raouf},
  journal={IEEE/ACM Transactions on networking},
  volume={20},
  number={1},
  pages={206--219},
  year={2011},
  publisher={IEEE}
}

@article{melo2013optimal,
  title={Optimal virtual network embedding: Node-link formulation},
  author={Melo, Marcio and Sargento, Susana and Killat, Ulrich and Timm-Giel, Andreas and Carapinha, Jorge},
  journal={IEEE Transactions on Network and Service Management},
  volume={10},
  number={4},
  pages={356--368},
  year={2013},
  publisher={IEEE}
}

@article{rost2019virtual,
  title={Virtual network embedding approximations: Leveraging randomized rounding},
  author={Rost, Matthias and Schmid, Stefan},
  journal={IEEE/ACM Transactions on Networking},
  volume={27},
  number={5},
  pages={2071--2084},
  year={2019},
  publisher={IEEE}
}

@article{moura2018branch,
  title={A branch-and-price algorithm for the single-path virtual network embedding problem},
  author={Moura, Leonardo FS and Gaspary, Luciano P and Buriol, Luciana S},
  journal={Networks},
  volume={71},
  number={3},
  pages={188--208},
  year={2018},
  publisher={Wiley Online Library}
}

@article{mijumbi2015path,
  title={A path generation approach to embedding of virtual networks},
  author={Mijumbi, Rashid and Serrat, Joan and Gorricho, Juan-Luis and Boutaba, Raouf},
  journal={IEEE Transactions on Network and Service Management},
  volume={12},
  number={3},
  pages={334--348},
  year={2015},
  publisher={IEEE}
}

@article{jarray2014decomposition,
  title={Decomposition approaches for virtual network embedding with one-shot node and link mapping},
  author={Jarray, Abdallah and Karmouch, Ahmed},
  journal={IEEE/ACM Transactions on Networking},
  volume={23},
  number={3},
  pages={1012--1025},
  year={2014},
  publisher={IEEE}
}

@article{rost2019parametrized,
  title={Parametrized complexity of virtual network embeddings: Dynamic \& linear programming approximations},
  author={Rost, Matthias and D{\"o}hne, Elias and Schmid, Stefan},
  journal={ACM SIGCOMM Computer Communication Review},
  volume={49},
  number={1},
  pages={3--10},
  year={2019},
  publisher={ACM New York, NY, USA}
}

@article{lu2006efficient,
  title={Efficient mapping of virtual networks onto a shared substrate},
  author={Lu, Jing and Turner, Jonathan},
  year={2006}
}

@article{yu2008rethinking,
  title={Rethinking virtual network embedding: Substrate support for path splitting and migration},
  author={Yu, Minlan and Yi, Yung and Rexford, Jennifer and Chiang, Mung},
  journal={ACM SIGCOMM Computer Communication Review},
  volume={38},
  number={2},
  pages={17--29},
  year={2008},
  publisher={ACM New York, NY, USA}
}

@article{cheng2011virtual,
  title={Virtual network embedding through topology-aware node ranking},
  author={Cheng, Xiang and Su, Sen and Zhang, Zhongbao and Wang, Hanchi and Yang, Fangchun and Luo, Yan and Wang, Jie},
  journal={ACM SIGCOMM Computer Communication Review},
  volume={41},
  number={2},
  pages={38--47},
  year={2011},
  publisher={ACM New York, NY, USA}
}

@inproceedings{gong2014toward,
  title={Toward profit-seeking virtual network embedding algorithm via global resource capacity},
  author={Gong, Long and Wen, Yonggang and Zhu, Zuqing and Lee, Tony},
  booktitle={IEEE INFOCOM 2014-IEEE Conference on Computer Communications},
  pages={1--9},
  year={2014},
  organization={IEEE}
}

@article{zhang2016virtual,
  title={Virtual network embedding based on the degree and clustering coefficient information},
  author={Zhang, Peiying and Yao, Haipeng and Liu, Yunjie},
  journal={IEEE Access},
  volume={4},
  pages={8572--8580},
  year={2016},
  publisher={IEEE}
}

@article{zhang2017virtual,
  title={Virtual network embedding based on computing, network, and storage resource constraints},
  author={Zhang, Peiying and Yao, Haipeng and Liu, Yunjie},
  journal={IEEE Internet of Things Journal},
  volume={5},
  number={5},
  pages={3298--3304},
  year={2017},
  publisher={IEEE}
}

@inproceedings{infuhr2013solving,
  title={Solving the virtual network mapping problem with construction heuristics, local search and variable neighborhood descent},
  author={Inf{\"u}hr, Johannes and Raidl, G{\"u}nther R},
  booktitle={European Conference on Evolutionary Computation in Combinatorial Optimization},
  pages={250--261},
  year={2013},
  organization={Springer}
}

@inproceedings{rguez2025grasp,
  title={A grasp-based approach for dynamic and reliable virtual network embedding},
  author={Rguez, Amine and Hadjadj-Aoul, Yassine and Rubino, Gerardo},
  booktitle={2025 Global Information Infrastructure and Networking Symposium (GIIS)},
  pages={1--6},
  year={2025},
  organization={IEEE}
}

@article{zhu2016modified,
  title={A modified ACO algorithm for virtual network embedding based on graph decomposition},
  author={Zhu, Fangjin and Wang, Hua},
  journal={Computer Communications},
  volume={80},
  pages={1--15},
  year={2016},
  publisher={Elsevier}
}

@INPROCEEDINGS{zhu2006algorithms,
  author={Zhu, Y. and Ammar, M.},
  booktitle={Proceedings IEEE INFOCOM 2006. 25TH IEEE International Conference on Computer Communications}, 
  title={Algorithms for Assigning Substrate Network Resources to Virtual Network Components}, 
  year={2006},
  volume={},
  number={},
  pages={1-12},
  doi={10.1109/INFOCOM.2006.322}
}

@article{zhang2013unified,
  title={A unified enhanced particle swarm optimization-based virtual network embedding algorithm},
  author={Zhang, Zhongbao and Cheng, Xiang and Su, Sen and Wang, Yiwen and Shuang, Kai and Luo, Yan},
  journal={International Journal of Communication Systems},
  volume={26},
  number={8},
  pages={1054--1073},
  year={2013},
  publisher={Wiley Online Library}
}

@article{zhang2020vne,
  title={VNE-HPSO: Virtual network embedding algorithm based on hybrid particle swarm optimization},
  author={Zhang, Peiying and Hong, Yanrong and Pang, Xue and Jiang, Chunxiao},
  journal={IEEE Access},
  volume={8},
  pages={213389--213400},
  year={2020},
  publisher={IEEE}
}

@article{song2019divide,
  title={A divide-and-conquer evolutionary algorithm for large-scale virtual network embedding},
  author={Song, An and Chen, Wei-Neng and Gong, Yue-Jiao and Luo, Xiaonan and Zhang, Jun},
  journal={IEEE Transactions on Evolutionary Computation},
  volume={24},
  number={3},
  pages={566--580},
  year={2019},
  publisher={IEEE}
}

@inproceedings{fajjari2011vne,
  title={VNE-AC: Virtual network embedding algorithm based on ant colony metaheuristic},
  author={Fajjari, Ilhem and Aitsaadi, Nadjib and Pujolle, Guy and Zimmermann, Hubert},
  booktitle={2011 IEEE international conference on communications (ICC)},
  pages={1--6},
  year={2011},
  organization={IEEE}
}

@article{zhang2019virtual,
  title={Virtual network embedding based on modified genetic algorithm},
  author={Zhang, Peiying and Yao, Haipeng and Li, Maozhen and Liu, Yunjie},
  journal={Peer-to-Peer Networking and Applications},
  volume={12},
  number={2},
  pages={481--492},
  year={2019},
  publisher={Springer}
}

@article{liu2016optimal,
  title={Optimal virtual network embedding based on artificial bee colony},
  author={Liu, Xu and Zhang, Zhongbao and Li, Ximing and Su, Sen},
  journal={EURASIP Journal on Wireless Communications and Networking},
  volume={2016},
  number={1},
  pages={273},
  year={2016},
  publisher={Springer}
}

@article{yao2018novel,
  title={A novel reinforcement learning algorithm for virtual network embedding},
  author={Yao, Haipeng and Chen, Xu and Li, Maozhen and Zhang, Peiying and Wang, Luyao},
  journal={Neurocomputing},
  volume={284},
  pages={1--9},
  year={2018},
  publisher={Elsevier}
}

@inproceedings{dolati2019deepvine,
  title={DeepViNE: Virtual network embedding with deep reinforcement learning},
  author={Dolati, Mahdi and Hassanpour, Seyedeh Bahereh and Ghaderi, Majid and Khonsari, Ahmad},
  booktitle={IEEE INFOCOM 2019-IEEE Conference on Computer Communications Workshops (INFOCOM WKSHPS)},
  pages={879--885},
  year={2019},
  organization={IEEE}
}

@article{yan2020automatic,
  title={Automatic virtual network embedding: A deep reinforcement learning approach with graph convolutional networks},
  author={Yan, Zhongxia and Ge, Jingguo and Wu, Yulei and Li, Liangxiong and Li, Tong},
  journal={IEEE Journal on Selected Areas in Communications},
  volume={38},
  number={6},
  pages={1040--1057},
  year={2020},
  publisher={IEEE}
}

@inproceedings{habibi2020accelerating,
  title={Accelerating virtual network embedding with graph neural networks},
  author={Habibi, Farzad and Dolati, Mahdi and Khonsari, Ahmad and Ghaderi, Majid},
  booktitle={2020 16th International Conference on Network and Service Management (CNSM)},
  pages={1--9},
  year={2020},
  organization={IEEE}
}

@article{haeri2017virtual,
  title={Virtual network embedding via Monte Carlo tree search},
  author={Haeri, Soroush and Trajkovi{\'c}, Ljiljana},
  journal={IEEE transactions on cybernetics},
  volume={48},
  number={2},
  pages={510--521},
  year={2017},
  publisher={IEEE}
}

@article{orlowski2010sndlib,
  title={SNDlib 1.0—Survivable network design library},
  author={Orlowski, Sebastian and Wess{\"a}ly, Roland and Pi{\'o}ro, Michal and Tomaszewski, Artur},
  journal={Networks: An International Journal},
  volume={55},
  number={3},
  pages={276--286},
  year={2010},
  publisher={Wiley Online Library}
}

@article{knight2011internet,
  title={The internet topology zoo},
  author={Knight, Simon and Nguyen, Hung X and Falkner, Nickolas and Bowden, Rhys and Roughan, Matthew},
  journal={IEEE Journal on Selected Areas in Communications},
  volume={29},
  number={9},
  pages={1765--1775},
  year={2011},
  publisher={IEEE}
}

@inproceedings{zegura1996model,
  title={How to model an internetwork},
  author={Zegura, Ellen W and Calvert, Kenneth L and Bhattacharjee, Samrat},
  booktitle={Proceedings of IEEE INFOCOM'96. Conference on Computer Communications},
  volume={2},
  pages={594--602},
  year={1996},
  organization={IEEE}
}

@inproceedings{sanders2013think,
  title={Think locally, act globally: Highly balanced graph partitioning},
  author={Sanders, Peter and Schulz, Christian},
  booktitle={International Symposium on Experimental Algorithms},
  pages={164--175},
  year={2013},
  organization={Springer}
}

@article{raack2011cut,
  title={On cut-based inequalities for capacitated network design polyhedra},
  author={Raack, Christian and Koster, Arie MCA and Orlowski, Sebastian and Wess{\"a}ly, Roland},
  journal={Networks},
  volume={57},
  number={2},
  pages={141--156},
  year={2011},
  publisher={Wiley Online Library}
}

@article{caprara2016solving,
  title={Solving the temporal knapsack problem via recursive Dantzig--Wolfe reformulation},
  author={Caprara, Alberto and Furini, Fabio and Malaguti, Enrico and Traversi, Emiliano},
  journal={Information Processing Letters},
  volume={116},
  number={5},
  pages={379--386},
  year={2016},
  publisher={Elsevier}
}

@article{bezanson2017julia,
  title={Julia: A fresh approach to numerical computing},
  author={Bezanson, Jeff and Edelman, Alan and Karpinski, Stefan and Shah, Viral B},
  journal={SIAM review},
  volume={59},
  number={1},
  pages={65--98},
  year={2017},
  publisher={SIAM}
}

@article{dunning2017jump,
  title={JuMP: A modeling language for mathematical optimization},
  author={Dunning, Iain and Huchette, Joey and Lubin, Miles},
  journal={SIAM review},
  volume={59},
  number={2},
  pages={295--320},
  year={2017},
  publisher={SIAM}
}

@incollection{desrosiers2005primer,
  title={A primer in column generation},
  author={Desrosiers, Jacques and L{\"u}bbecke, Marco E},
  booktitle={Column generation},
  editor={Springer},
  pages={1--32},
  year={2005},
  publisher={Springer}
}

@article{karypis1997metis,
  title={METIS: A software package for partitioning unstructured graphs, partitioning meshes, and computing fill-reducing orderings of sparse matrices},
  author={Karypis, George and Kumar, Vipin},
  year={1997}
}

@inproceedings{pessoa2013out,
  title={In-out separation and column generation stabilization by dual price smoothing},
  author={Pessoa, Artur and Sadykov, Ruslan and Uchoa, Eduardo and Vanderbeck, Francois},
  booktitle={International Symposium on Experimental Algorithms},
  pages={354--365},
  year={2013},
  organization={Springer}
}

@article{sabidussi1966centrality,
  title={The centrality index of a graph},
  author={Sabidussi, Gert},
  journal={Psychometrika},
  volume={31},
  number={4},
  pages={581--603},
  year={1966},
  publisher={Springer-Verlag}
}

@InProceedings{Wang2025PriceBranch,
    author="Wang, Ze
    and S{\"u}{\ss}, Tim
    and Popov, Nikolay
    and Nagel, Lars",
    editor="Krejca, Martin S.
    and Wagner, Markus",
    title="Price-and-Branch Heuristic for Vector Bin Packing",
    booktitle="Evolutionary Computation in Combinatorial Optimization",
    year="2025",
    publisher="Springer Nature Switzerland",
    address="Cham",
    pages="249--265",
    isbn="978-3-031-86849-8"
}

@article{ahuja1988network,
  title={Network flows},
  author={Ahuja, Ravindra K and Magnanti, Thomas L and Orlin, James B},
  year={1988},
  publisher={Cambridge, Mass.: Alfred P. Sloan School of Management, Massachusetts~…}
}

@misc{Graphs2021,
  author       = {Fairbanks, James and Besan{\c{c}}on, Mathieu and Simon, Sch{\"o}lly and Hoffiman, J{\'u}lio and Eubank, Nick and Karpinski, Stefan},
  title        = {JuliaGraphs/Graphs.jl: an optimized graphs package for the Julia programming language},
  year         = 2021,
  url = {https://github.com/JuliaGraphs/Graphs.jl/}
}

@article{floyd1962algorithm,
  title={Algorithm 97: shortest path},
  author={Floyd, Robert W},
  journal={Communications of the ACM},
  volume={5},
  number={6},
  pages={345--345},
  year={1962},
  publisher={ACM New York, NY, USA}
}

@misc{schneider2025vnelib,
  author       = {Schneider, Alexis},
  title        = {VNE Lib: a library for offline Virtual Network Embedding instances},
  year         = 2021,
  url = {https://github.com/VNE-lib/}
}

@inproceedings{da2020impact,
  title={On the impact of novel function mappings, sharing policies, and split settings in network slice design},
  author={da Silva Coelho, Wesley and Benhamiche, Amal and Perrot, Nancy and Secci, Stefano},
  booktitle={2020 16th international conference on network and service management (CNSM)},
  pages={1--9},
  year={2020},
  organization={IEEE}
}

@article{benhamiche2026polytope,
  title={On the Virtual Network Embedding polytope},
  author={Benhamiche, Amal and Fouilhoux, Pierre and L{\'e}tocart, Lucas and Perrot, Nancy and Schneider, Alexis},
  journal={arXiv preprint arXiv:2601.11419},
  year={2026}
}

\end{document}